\documentclass[11pt]{article}

\usepackage{authblk}
\usepackage{graphicx}
\usepackage{amsfonts}
\usepackage{amsmath}
\usepackage{amsthm}
\usepackage{latexsym}
\usepackage{amssymb}
\usepackage{color}
\usepackage{verbatim}
\usepackage{xspace}
\usepackage{nicefrac}
\usepackage[top=1in, bottom=1in, left=1in, right=1in]{geometry}
\usepackage{algorithm}
\usepackage{algpseudocode}

\newtheorem{theorem}{Theorem}[section]
\newtheorem{lemma}[theorem]{Lemma}

\newtheorem{definition}[theorem]{Definition}

\newtheorem{claim}{Claim}
\newtheorem{subclaim}{Subclaim}

\algnewcommand\algorithmicinput{\textbf{Input:}}
\algnewcommand\INPUT{\item[\algorithmicinput]}
\algnewcommand\algorithmicoutput{\textbf{Output:}}
\algnewcommand\OUTPUT{\item[\algorithmicoutput]}

\def\wp{0.62}
\def\gap{1mm}
\def\cN{circular $\mathrm{N}$}
\def\fN{fuzzy $\mathrm{N}$}
\def\z6{$\mathrm{Z_6}$}
\def\zz{\mathrm{Z}}

\newcommand{\cG}{\mathcal{G}}

\newcommand{\cP}{\mathcal{P}}

\newcommand{\cU}{\mathcal{U}}

\newcommand{\al}{\ensuremath{\mathcal{L}}}

\title{Space complexity of list $H$-coloring revisited: the case of oriented trees\footnote{This research was supported by NSERC.}}
\author{L\'{a}szl\'{o} Egri}
\affil{Simon Fraser University}

\begin{document}

\maketitle

\begin{abstract}
	Digraphs $H$ for which the list homomorphism problem with template $H$ (LHOM($H$)) is in logspace ($\mathrm{L}$) was characterized by Egri et al.\ (SODA 2014): LHOM($H$) is in $\mathrm{L}$ if and only if $H$ does not contain a circular $\mathrm{N}$.\footnote{Assuming $\mathrm{L} \neq \mathrm{NL}$.} \emph{Undirected} graphs for which LHOM($H$) is in $\mathrm{L}$ can be characterized in terms forbidden induced subgraphs, and also via a simple inductive construction (Egri et al., STACS 2010). As a consequence, the logspace algorithm in the undirected case is simple and easy to understand. No such forbidden subgraph or inductive characterization, and no such simple and easy-to-understand algorithm is known in the case of digraphs. In this paper, \emph{in the case of oriented trees}, we refine and strengthen the results of Egri et al.\ (SODA 2014): we give a characterization of oriented trees $T$ for which LHOM($T$) is in $\mathrm{L}$ both in terms of forbidden induced subgraphs, and also via a simple inductive construction. Using this characterization, we obtain a simple and easy-to-analyze logspace algorithm for LHOM($T$). We also show how these oriented trees can be recognized in time $O(|V(T)|^3)$ (the straightforward implementation of the algorithm given in SODA 2014 runs in time  $O(|V(H)|^8)$ for oriented trees). An algebraic characterization of these trees is also provided.
\end{abstract}
\section{Introduction}

Given two digraphs $G$ and $H$, a homomorphism $\varphi : G \rightarrow H$ is a mapping $\varphi : V (G) \rightarrow V (H)$ such that $uv \in A(G)$ implies that $\varphi(u) \varphi(v) \in A(H)$, where $A(G)$ denotes the arc set of $G$. The corresponding algorithmic problem \emph{Digraph Homomorphism} asks if $G$ has a homomorphism to $H$. For example, it is easy to see that $G$ has a homomorphism into the clique $K_c$ if and only if G is $c$-colorable. Instead of digraphs, one can consider homomorphism problems in the more general context of relational structures. Feder and Vardi \cite{Feder93:monotone} observed that the standard framework for the Constraint Satisfaction Problem (CSP) can be formulated as homomorphism problems for relational structures. In fact, they showed that every such problem is equivalent to a Digraph Homomorphism problem, hence Digraph Homomorphism is as expressive as the CSP in general.

The expressive power of Digraph Homomorphism can be increased by introducing lists. Given digraphs $G$ and $H$ and a list $L(v) \subseteq V(H)$ for each $v \in V(G)$, a \emph{list homomorphism} $\varphi$ from $G$ to $H$ is a homomorphism from $G$ to $H$ such that $\varphi(v) \in L(v)$ for each $v \in V(G)$. The \emph{List Homomorphism problem with template $H$} (LHOM($H$)) is the following algorithmic problem. Given a digraph $G$ and a list $L(v)$ for each $v \in V(G)$, decide if there is a list homomorphism from $G$ to $H$.\footnote{We remark that LHOM$(H)$ is identical to CSP($\mathbb{B}$), where $\mathbb{B}$ is the relational structure that contains the binary relation that is the arc set of the digraph $H$, and a unary relation $U_S = S$ for each $S \subseteq V(H)$.} The List Homomorphism problem was introduced by Feder and Hell in 1998 \cite{FH:98:LHR}, and it has been studied since then extensively \cite{stacs_lhom,Feder/et_al:07:LHG,Feder/et_al:99:LHC,FederHH03,Gutin06:mincosthomomorphism,Hell/Rafiey:11:DLH}.

In this paper, we study List Homomorphism problems of logarithmic space complexity. Such problems with graph and digraph templates have been studied in \cite{stacs_lhom,soda_lhom,lics_lhom}. When $H$ is an undirected graph, LHOM($H$) is in logspace\footnote{These results assume that $\mathrm{L} \neq \mathrm{NL}$.} ($\mathrm{L}$) if and only if $H$ does not contain any of a certain set of graphs as an induced subgraph, or equivalently, if $H$ can be inductively constructed using two simple operations (see \cite{stacs_lhom} for details). We call these graphs \emph{skew decomposable}. Relying on this inductive construction, the logspace algorithm for LHOM($H$) is remarkably simple. When $H$ is a \emph{digraph}, LHOM($H$) is in $\mathrm{L}$ if and only if $H$ does not contain a so-called \emph{circular $\mathrm{N}$} (see \cite{soda_lhom}). We note that $H$ not containing a circular $\mathrm{N}$ is \emph{not} a forbidden induced subgraph characterization. Although the lack of a circular $\mathrm{N}$ in $H$ gives a unifying reason why LHOM($H$) is in $\mathrm{L}$, the first of the two existing logspace algorithms for LHOM($H$) is quite complicated (see \cite{soda_lhom}), and the second algorithm (\emph{a symmetric Datalog program}) requires an involved and technical analysis, which is the subject of \cite{lics_lhom}. This is in stark contrast to the simple and easy-to-analyze logspace algorithm for LHOM($H$) when $H$ is a skew decomposable undirected graph \cite{stacs_lhom}. Whether digraphs for which LHOM($H$) is in $\mathrm{L}$ can be characterized in terms of forbidden induced subgraphs, whether they enjoy an inductive construction, and whether such a construction could be used to build a simple logspace algorithm for LHOM($H$) are intriguing open problems. In this paper, we answer these questions for oriented trees in the positive. (We remark that our results re-prove the $\mathrm{L}-\mathrm{NL}$ dichotomy for LHOM($T$) when $T$ is an oriented tree.)

Graphs and digraphs $H$ for which LHOM($H$) is in $\mathrm{L}$ form a natural class, and we believe that our inductive characterization of such oriented trees could be of independent interest. In fact, the FPT-algorithm in \cite{esa_dlhom} uses the inductive characterization of skew decomposable graphs in an essential way.

\textbf{Detailed results and structure of the paper:} In Section~\ref{preliminaries}, we introduce basic concepts and define an oriented path we call \z6, and a class of oriented paths we call fuzzy $\mathrm{N}$-s. In Section~\ref{structural}, we describe a way to construct oriented trees inductively (Definition~\ref{constructible}). We proceed to prove the main combinatorial result of the paper: an oriented tree $T$ contains a $\mathrm{Z_6}$ or a fuzzy $\mathrm{N}$ as an induced subgraph if and only if $T$ can be constructed using the inductive construction in Definition~\ref{constructible} (Theorem~\ref{construction_theorem}). In Section~\ref{algorithm}, we briefly argue that if a tree $T$ contains a $\mathrm{Z_6}$ or a fuzzy $\mathrm{N}$ as an induced subgraph, then LHOM($T$) is $\mathrm{NL}$-hard. Then for oriented trees $T$ that do not contain any such induced subgraph, we provide a simple logspace algorithm relying on the aforementioned inductive characterization of $T$. In Section~\ref{equivalence}, we give an \emph{unconditional} proof that an oriented tree contains a circular $\mathrm{N}$ if and only if $T$ contains a $\mathrm{Z_6}$ or a fuzzy $\mathrm{N}$ as an induced subgraph. (Note that if we assume that $\mathrm{L} \neq \mathrm{NL}$, then there is a simpler argument.) In Section~\ref{algebra_sec}, we show that $T$ does not contain a $\mathrm{Z_6}$ or a fuzzy $\mathrm{N}$ as an induced subgraph if and only if $T$ admits a \emph{Hagemann-Mitschke chain of conservative polymorphisms of length $3$} (note that this is similar to the characterization of undirected graphs in \cite{stacs_lhom}). We also give an example of a digraph in Section~\ref{algebra_sec} that admits a Hagemann-Mitschke chain of conservative polymorphisms of length $n+1$ but not of length $n$ (Theorem~\ref{ladder_thm}). Therefore we can conclude that in this respect, general digraphs behave differently from oriented trees. In Section~\ref{faster}, we give a $O(|V(T)|^3)$ algorithm to recognize oriented trees that do not contain a $\mathrm{Z_6}$ or a fuzzy $\mathrm{N}$ as an induced subgraph. This is significantly faster than the $O(|V(T)|^8)$-time implementation of the recognition algorithm in \cite{soda_lhom} restricted oriented tree inputs.\footnote{Note that $O(|V(T)|^8)$ is the running time of the straightforward implementation of the recognition algorithm (in \cite{soda_lhom}) when inputs are assumed to be oriented trees. We also note that this algorithm runs faster on trees than on general digraphs. We made no attempt to improve the running time of this algorithm. See Appendix A.} 

For the sake of completeness, Appendix A contains both an inductive construction and an ``explicit'' characterization of oriented paths that contain no \z6\ or \fN\ as an induced subgraph.

We summarize the main results of this paper in Theorem~\ref{main}. Note that some parts of this theorem come from \cite{soda_lhom}, as explained below.
\begin{theorem}\label{main}
	Let $T$ be an oriented tree. Then the following conditions are equivalent:
	\begin{enumerate}
		\item $T$ contains no induced subgraph that is a \z6\ or a \fN;
		\item $T$ can be constructed inductively as in Definition~\ref{constructible};
		\item $T$ contains no circular $\mathrm{N}$;
		\item $T$ admits a chain of conservative Hagemann-Mitschke polymorphisms of length $3$;		
		\item $T$ admits a chain of conservative Hagemann-Mitschke polymorphisms of length $n$, for some $n \geq 1$.
	\end{enumerate}
	If the above conditions hold, then $\mathrm{LHOM(T)}$ is in $\mathrm{L}$. Otherwise $\mathrm{LHOM(T)}$ is $\mathrm{NL}$-hard.
\end{theorem}

The outline of the proof of this theorem is the following:\footnote{In Appendix B, we give direct proofs that (2) $\Rightarrow$ (1) and (1) $\Rightarrow$ (3), and we also give a not direct but short proof that (4) $\Rightarrow$ (1).}
\begin{itemize}
\item (1) $\Rightarrow$ (2) is the content of Lemma~\ref{d1}.
\item (2) $\Rightarrow$ (4) is the content of Lemma~\ref{HM_chain_defined}.
\item (4) $\Rightarrow$ (5) is trivial.
\item (5) $\Rightarrow$ (3) is by \cite{soda_lhom}.
\item (3) $\Rightarrow$ (1) is the content of Lemma~\ref{z->cN}.
\end{itemize}
It is shown in \cite{soda_lhom} that if $T$ contains no \cN\ then LHOM($T$) is in $\mathrm{L}$, and otherwise LHOM($T$) is $\mathrm{NL}$-hard. However, as discussed above, the merit of the logspace algorithm in this paper is its simple inductive nature.
\section{Preliminaries}\label{preliminaries}

\subsection{Digraphs and related concepts}
Let $G$ be a digraph. (All digraphs in this paper are finite.) An arc of $G$ from vertex $a$ to vertex $b$ is denoted by $ab$. We call $a$ and $b$ the \emph{endpoints} of $ab$. If we want to be more specific, we call $a$ the \emph{tail} and $b$ the \emph{head} of $ab$. If $v \in V(G)$ we call $u$ an \emph{inneighbour} (\emph{outneighbours}) of $v$ if $uv \in A(G)$ ($vu \in A(G)$). The \emph{indegree (outdegree)} of a vertex is the number of its inneighbours (outneighbours). A digraph $G$ is \emph{connected} if the undirected graph obtained from $G$ by replacing arcs with undirected edges (the \emph{underlying undirected graph}) is connected. For a disconnected digraph $G$, a \emph{component} of $G$ is a maximal subgraph that is connected.

If $V(G)$ can be partitioned into non-empty sets called \emph{vertex levels} $L_0,\dots,L_n$ such that for each arc $ab$, $a \in L_i$ and $b \in L_{i+1}$ for some $i < n$, then we call $G$ \emph{leveled}. Observe that if a connected digraph is leveled, then the partition $L_0,\dots,L_n$ is unique. We call $L_0$ the \emph{bottom} vertex level and $L_n$ the \emph{top} vertex level. For a vertex $v \in V$, we say that \emph{$v$ is in the bottom (top) level if $v \in L_0$ ($v \in L_n$)}. Given a vertex $v$ of $G$, we use $\ell(v)$ to denote the index such that $v \in L_{\ell(v)}$. For $0 \leq i \leq n - 1$ , we denote the set of arcs of $G$ with one endpoint in $L_i$ and the other one in $L_{i+1}$ with $\al_i$. We call $\al_i$ the \emph{$i$-th arc level} of $G$. We define \emph{bottom} and \emph{top arc levels} in the natural way. If we need to be explicit, we write $L^G_i$ ($\al_i^G$) instead of $L_i$ ($\al_i$) to mean the $i$-th vertex (arc) level of digraph $G$. Note that we can think of $\al_i$ as a digraph, and we will often do so without explicitly mentioning. Also note that if $u$ is a vertex in $L_i$ or $L_{i+1}$ such that $u$ is not the endpoint of an arc in $\al_i$ then $u$ \emph{does not} belong to the digraph $\al_i$. Given an arc $a$ of $G$, we use $\ell(a)$ to denote the index such that $a \in \al_{\ell(a)}$. (Note that using $\ell$ for both vertex and arc levels will cause no confusion.)

Given a digraph $G$, $r(G)$ denotes the digraph obtained from $G$ by replacing every arc $ab$ with $ba$. An \emph{oriented walk} $W$ is a sequence of vertices $a_1 a_2 \dots a_m$, where precisely one of $a_i a_{i+1}$ or $a_{i+1} a_i$ is an arc. Arc $a_i a_{i+1}$ is called a \emph{forward} arc, and $a_{i+1} a_i$ a \emph{backward} arc. Let $W = a_1 a_2 \dots a_m$ be an oriented walk. An \emph{oriented path} is a \emph{simple} oriented walk, i.e., each vertex in the walk appears only once. Suppose that $e$ is an arc of $W$ with endpoints $a_j$ and $a_{j+1}$, and $e'$ is an arc of $W$ with endpoints $a_{k-1}$ and $a_k$ (note that both $e$ or $e'$ could be backward or forward). Then $W(a_j,a_k)$ (a walk from vertex $a_j$ to $a_k$), $W(e,a_k)$ (a walk from arc $e$ to vertex $a_k$), $W(a_j,e')$ (a walk from a vertex $a_j$ to arc $e'$) and $W(e,e')$ (a walk from arc $e$ to arc $e'$) all denote the subwalk $a_ja_{j+1}\dots a_k$ of $W$. For an oriented path $W = a_1 a_2 \dots a_m$, there is a natural total order $\preceq$ on its vertices, i.e., $a_i \preceq a_j$ if and only if $i \leq j$. This order helps us refer to parts of $W$: the \emph{first} and \emph{last} vertices of $W$ are $a_1$ and $a_n$, respectively. A vertex $a_i$ of $W$ is \emph{before} (\emph{after}) $a_j$ if $i \preceq j$ ($j \preceq i$). We use $\bar{W}$ to denote the path that is isomorphic to $W$, but the order associated with the path is reversed, i.e., $\bar{W} = a_m a_{m-1} \dots a_1$. If $P=a_1\dots a_n$ and $Q = b_1 \dots b_m$ are two oriented paths, then $PQ$ is the \emph{concatenation} of $P$ and $Q$, i.e., the oriented path $P=a_1\dots a_nb_2 \dots b_m$, where we identify the last vertex of $P$, $a_n$, and the first vertex of $Q$, $b_1$ (i.e., the arc on vertices $a_n$ and $b_2$ is $a_nb_2$ if $b_1b_2$ is an arc of $Q$, and it is $b_2a_n$ if $b_2b_1$ is an arc of $Q$). The \emph{height} of an oriented walk $W$, denoted by \emph{$height(W)$}, is the number of different vertex levels in which $W$ contains at least one vertex minus $1$. If $\ell(a_1) < \ell(a_n)$, then we say $W$ is an \emph{upward} walk, and if $\ell(a_1) > \ell(a_n)$, then we say that $P$ is a \emph{downward} walk. (When $\ell(a_1) = \ell(a_n)$, the walk is neither upward nor downward.) The \emph{net length} of an oriented path is the number of forward arcs minus the number of backward arcs in the walk, and it is denoted by $net(W)$.

An \emph{oriented tree} is a digraph such that the underlying undirected graph is a tree. Observe that an oriented tree $T$ is always leveled. Furthermore, let $a$ be a vertex or an arc of $T$, and let $b$ also be a vertex or an arc of $T$. Then observe that since $T$ is a tree, the oriented path $P(a,b)$ is unique.

In what follows, when we say that digraph $X$ is \emph{of the form} $Y$, where $Y$ is digraph, we mean that there is an isomorphism between $X$ and $Y$. Similarly, saying that digraph $X$ \emph{is} a $Y$ means $X$ is isomorphic to $Y$.

\begin{definition}\label{Z_6}
	$\mathrm{Z}_i^s$, where $1 \leq i$ and $s \in \{f=0,f=1,l=0,l=1\}$, is used to denote oriented paths of the following form. $\mathrm{Z}_i^s$ is of the form $a_1\dots a_i$, where if $a_ja_{j+1}$ is a forward (backward) arc, then $a_{j+1}a_{j+2}$ is a backward (forward arc) arc. Observe that it is always the case that $height(\mathrm{Z}_i^s)=1$, and therefore $\mathrm{Z}_i^s$ has two vertex levels, a bottom level $L_0$ and a top level $L_1$. The superscripts $f=0$ and $f=1$ stand for the first vertex $a_1$ being in $L_0$ and $L_1$, respectively. Similarly, the superscripts $l=0$ and $l=1$ stand for the last vertex $a_i$ being in $L_0$ and $L_1$, respectively. $\mathrm{Z}_i$ stands for either $\mathrm{Z}_i^{f=0}$ or $\mathrm{Z}_i^{f=1}$. Observe that $\mathrm{Z}_1$ is a single vertex. $\mathrm{Z}$ stands for $\mathrm{Z}_i$ for some $i$. A $\mathrm{Z}_6^{f=0}$ can be seen in Figure~\ref{f_constr}.
\end{definition}

\begin{definition}\label{def_fuzzyN}
	Let $n$ be a positive integer. We say that a digraph $P$ is a \emph{fuzzy path} if $P$ or $\bar{P}$ is of the form $P_1 P_2 \dots P_n$, where
	\begin{itemize}
		\item If $n = 1$, then $P_1$ is $\mathrm{Z}_i$, where $1\leq i \leq 5$.
		\item If $n \geq 2$, then $P_1$ is $\mathrm{Z}_i^{l=1}$ for some $i \leq 5$, and $P_n = \mathrm{Z}_j^{f=0}$ for some $j \leq 5$.
		\item For each $1 < i < n$, $P_i$ is either of the form $\mathrm{Z}_2^{f=0}$, or $\mathrm{Z}_4^{f=0}$.
	\end{itemize}
	A fuzzy path $P$ is \emph{minimal}, if it has only one vertex in both $L_0$ and $L_{height(P)}$.
	
	An oriented path $P$ is a \emph{fuzzy $\mathrm{N}$} if $P$ is of the form $P_1T\bar{P}_2BP_3$, where $P_1$, $P_2$ and $P_3$ are minimal fuzzy paths of the same height, and height at least $2$, $T$ is of the form $\mathrm{Z}_1$ or $\mathrm{Z}_3^{f=1}$, and $B$ is of the form $\mathrm{Z}_1$ or $\mathrm{Z}_3^{f=0}$.  A fuzzy $\mathrm{N}$ is illustrated in Figure~\ref{f_constr}.
\end{definition}

\begin{figure}[htb]
	\begin{center}
		\includegraphics[scale=\wp]{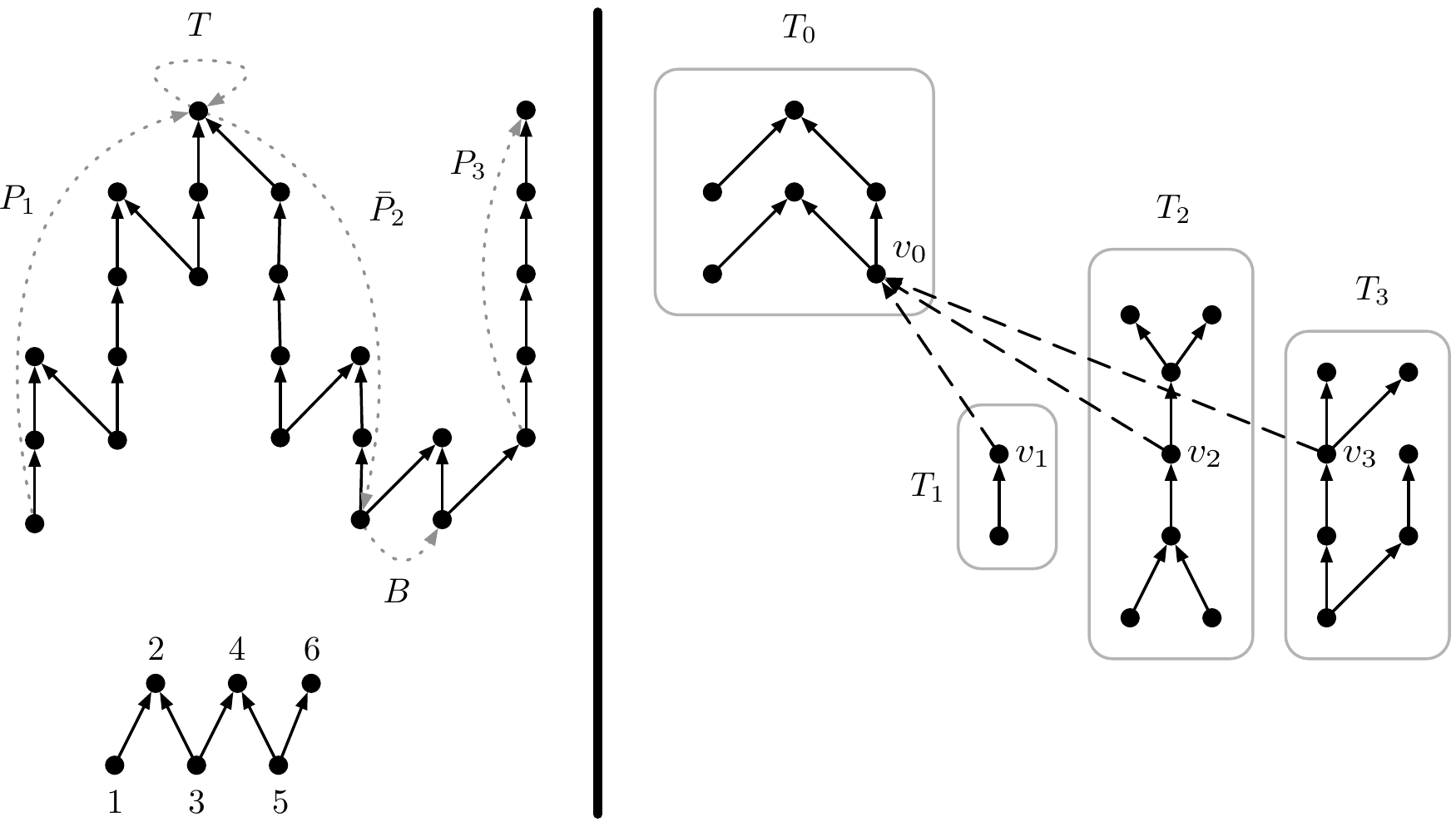}
	\end{center}
	\caption{A fuzzy $\mathrm{N}$ in which $T$ is a single vertex (top left), and a $\mathrm{Z}_6^{f=0}$ (bottom left). The up-join of $T_0,T_1,T_2,T_3$ (right).}\label{f_constr}
\end{figure}

\subsection{Homomorphisms and polymorphisms}

The List Homomorphism problem LHOM($H$) has already been defined in the Introduction. Let $G$ and $H$ be leveled connected digraphs. Notice that any (list)-homomorphism $h$ from $G$ to $H$ must be \emph{level-preserving}:
\begin{itemize}
	\item if $L^G$ is a vertex level of $G$, then $h(L^G) \subseteq L^H$, where $L^H$ is some vertex level of $H$, and
	\item if $u \in L^G_i$ and $v \in L_{i'}^G$, where $L_i^G$ and $L_{i'}^G$ are some vertex levels of $G$ ($0 \leq i,i' \leq height(G)$), then $\ell(h(u)) - \ell(h(v)) = i - i'$.
\end{itemize}
This level-preserving property of homomorphisms will be used implicitly.

\begin{definition}\label{def-HM} Let $H$ be a digraph. An operation $f : V(H)^m \rightarrow V(H)$
	is a polymorphism of $H$ if $f(v_{11},v_{12},\dots,v_{1m})f(v_{21},v_{22},\dots,v_{2m}) \in A(H)$ whenever
	$v_{11}v_{21},v_{12}v_{22},\dots,v_{1m}v_{2m} \in A(H)$. Operation $f$ is \emph{conservative} if $f(v_1,\dots,v_m) \in \{v_1,\dots,v_m\}$. A sequence $f_1, \dots, f_k$ of ternary operations is called a \emph{Hagemann-Mitschke chain of length $k$} if it satisfies the identities
	\begin{itemize}
		\item $x = f_1(x,y,y)$
		\item $f_i(x,x,y) = f_{i+1}(x,y,y)$ for all $i=1, \dots, k-1$
		\item $f_k(x,x,y)=y.$
	\end{itemize}
	We say that $H$ {\em admits} an HM-chain $f_1,f_2,\dots,f_k$ if each $f_i$ is a polymorphism of $H$.
\end{definition} 
\section{A structural characterization}\label{structural}

In this section, we define an inductive construction of oriented trees, and show that oriented trees that can be constructed this way are precisely the oriented trees that do not contain a $\mathrm{Z_6}$ or fuzzy $\mathrm{N}$ as an induced subgraph (Theorem~\ref{construction_theorem}).

\begin{definition}[Inductive construction]\label{constructible}
	\label{L_construction}
	Let $T_0,T_1,\dots,T_n$ be oriented trees.
	\begin{enumerate}
		\item Let $v_0$ be a vertex in the bottom vertex level of $T_0$;
		\item Let $v_i$ be a vertex of $T_i$, for each $1 \leq i \leq n$, either in the top (bottom) vertex level of $T_i$, or such that $v_i$ is the only vertex in $L_{\ell(v_i)}^{T_i}$ with out-degree (in-degree) greater than $0$.
	\end{enumerate}
	The \emph{up-join (down-join)} of $T_0,T_1,\dots,T_n$ is the oriented tree obtained by taking the disjoint union of $T_i$ for $1 \leq i \leq n$, and adding all arcs $v_iv_0$ ($v_0v_i$) for each $1 \leq i \leq n$. An example of this construction is given in Figure~\ref{f_constr}.
	
	We call $T_0$ the \emph{central tree}, $v_0$ the \emph{central vertex}, and $v_i$, where $1 \leq i \leq n$, the \emph{join vertices}. When we specify a list of trees and we take their up-join (down-join) the first tree in the list is always meant to be the central tree.
	
	If $G$ is an oriented tree with a single vertex, we say that $G$ is \emph{constructible}. Inductively, if $T_0,T_1,\dots,T_n$ are constructible, then their up-join and down-join (for some central and join vertices satisfying the conditions above) are also \emph{constructible}.
\end{definition}

The main structural result of this section is the following theorem.
\begin{theorem}\label{construction_theorem}
	An oriented tree is constructible if and only if $T$ contains neither $\mathrm{Z_6}$, nor a fuzzy $\mathrm{N}$ as an induced subgraph.
\end{theorem}

We need the following lemma a number of times to prove the existence of fuzzy $\mathrm{N}$'s.

\begin{lemma}\label{fuzzy_N_present}
	Let $P$ be an oriented path. Then $P$ contains vertices  $b_1,t_1,b_2,t_2$ with the following properties:
	\begin{enumerate}
		\item when we traverse $P$ from first vertex to last vertex, we encounter $b_1,t_1,b_2,t_2$ in this order,
		\item $b_1$ and $b_2$ are in $L_x$ for some $x$, and $t_1$ and $t_2$ are in $L_y$ for some $y$, where $y \geq x + 2$,
		\item no vertex of $P(b_1,t_2)$ is in level $L_{x-1}$ or level $L_{y+1}$,
	\end{enumerate}
	if and only if $P$ contains a fuzzy $\mathrm{N}$.
\end{lemma}
\begin{proof}
	Suppose that $P$ contains vertices $b_1,t_1,b_2,t_2$ with the above properties. Suppose that there is a path $Q$ among $P(b_1,t_1)$, $\bar{P}(b_2,t_1)$ or $P(b_2,t_2)$ such that $Q$ contains vertices $b_1',t_1',b_2',t_2'$ with same properties as $b_1,t_1,b_2,t_2$ in $P$, respectively. Then we work with $Q$ instead of $P$. Repeating this argument sufficiently many times, we can assume without loss of generality that each arc level of each subpath $P(b_1,t_1)$, $\bar{P}(b_2,t_1)$ and $P(b_2,t_2)$ of $P$ contains only one component of the given subpath. Since in addition, none of $P(b_1,t_1)$, $\bar{P}(b_2,t_1)$ and $P(b_2,t_2)$ contain $\mathrm{Z_6}$ as an induced subgraph, each of $P(b_1,t_1)$, $\bar{P}(b_2,t_1)$ and $P(b_2,t_2)$ satisfies the $3$ conditions in Definition~\ref{def_fuzzyN}. Therefore these subpaths are fuzzy. By throwing away unnecessary vertices, we can also assume that $b_1$ and $t_2$ are the only vertices of $P(b_1,t_1)$ and $P(b_2,t_2)$ in $L_x$ and $L_y$, respectively.
	
	Now it is easy to see that $P(b_1,t_2)$ is a \fN: let $t_1'$ be the first vertex of $P(b_1,t_1)$ in $L_y$, and let $t_1''$ the last vertex of $\bar{P}(b_2,t_1)$ in $L_y$. Let $T = P(t_1',t_1'')$. Similarly, let $b_2'$ the first vertex of $\bar{P}(b_2,t_1)$ in $L_x$, and let $b_2''$ be the last vertex of $P(b_2,t_2)$ in $L_x$. Let $B = P(b_2',b_2'')$. Then $P(b_1,t_1') T \bar{P}(b_2',t_1'') B P(b_2'',t_2)$ is a fuzzy $\mathrm{N}$.
	
	The converse is straightforward.
\end{proof}

\begin{lemma}\label{fuzzy_if_bt}
	Let $P$ be an oriented path containing no $\mathrm{Z_6}$ or fuzzy $\mathrm{N}$ as induced subgraphs. If the first vertex of $P$ is in the bottom level of $P$, and the last vertex of $P$ is in the top level of $P$, then $P$ is a fuzzy path.
\end{lemma}
\begin{proof}
Let the first and last vertices of $P$ be $s$ and $t$, respectively.
Assume for contradiction that $P$ contains a vertex $t_1'$ before a vertex $b_2'$ such that $\ell(t_1) \geq \ell(b_2) + 2$. Let $t_1$ be a vertex of $P(s,t_1')$ in $L_a$ where $a$ is maximal. Let $b_2$ be a vertex of $P(t_1',t)$ in $L_b$ such that $b$ is minimal.

Since $\ell(s) \leq \ell(b_2)$ and $\ell(t_1) \leq \ell(t)$, we can find $b_1$ (before $t_1$ and $t_2$ (after $b_2$) such that $b_1,t_1,b_2,t_2$ satisfy the conditions of Lemma~\ref{fuzzy_N_present}, and therefore $P$ contains a \fN\ as induced subgraph, a contradiction. Therefore if $P = a_1\dots a_n$, then for any $i < j$, $\ell(a_i) \leq
\ell(a_j) + 1$. Since $P$ does not contain \z6\ as an induced subgraph, we conclude that $P$ must be fuzzy.
\end{proof}

Lemma~\ref{technical_lemma} is the main technical result used in the proof of Theorem~\ref{construction_theorem}.

\begin{lemma}\label{technical_lemma}
	Let $T$ be an oriented tree. Assume that $T$ does not contain a $\mathrm{Z_6}$ or a fuzzy $\mathrm{N}$ as an induced subgraph. Then there is an arc level $\al$ of $T$ that contains at most one component.
\end{lemma}
\begin{proof}
	In this proof, $L_i$ and $\al_i$ always refer to vertex and arc levels of $T$, respectively. Similarly, the function $\ell(\cdot)$ gives the index of the vertex level or arc level of a vertex or arc of $T$, respectively. Choose two arbitrary vertices $b \in L_0$ and $t \in L_{height(T)}$. Let $S$ denote the (unique) path $P(b,t)$. Path $S$ is fixed for the rest of the proof. By Lemma~\ref{fuzzy_if_bt}, $S$ is a fuzzy. Since $S$ is a fuzzy path, there is a component $C_i$ of the digraph $\al_i$ (for all $0 \leq i \leq height(T)$), such that all arcs of $S$ in $\al_i$ belong to the component $C_i$. We say that an arc $d \in A(T)$ in $\al_i$ is \emph{separated from $S$} if $d$ does not belong to $C_i$. We can assume that there is a separated arc in $\al_i$ for each $0 \leq i \leq height(T)-1$, because otherwise we would have the desired property stated in the lemma. We use the existence of these separated arcs to obtain a contradiction.
	
	We will inductively fix a sequence of paths $F_0,\dots,F_q$ in $T$. For each $0 \leq i \leq q$, $A_i$ will denote the set of arcs of $F_i$ that are separated from $S$. To define $F_0$, we find a separated arc $a_0'$ in $\al_0$, and let $F_0'$ be the unique oriented path from $t$ (recall that $t$ is the last vertex of $S$ in $L_n$) to $a_0'$. Let $c_0$ be the last common vertex of $S$ and $F_0'$, and $a_0$ be the first separated arc of $F_0'(c_0,a_0')$ in $\al_0$. Then $F_0$ is the subpath $F_0'(c_0,a_0)$.
	
	Assuming that $F_0,\dots, F_{i-1}$ have been defined, we define $F_i$ inductively as follows. If for each $0 \leq j \leq height(T) - 1$, there is an index $j'$, such that $\al_j \cap A_{j'} \neq \emptyset$, that is, for each arc level of $T$, some $A_{j'}$ contains a separated arc of that arc level, then we set $q = i-1$, and the construction of the paths $F_0,\dots,F_q$ is completed. Otherwise, let $m$ be minimum such that there is no (separated) arc in $A_0 \cup \dots \cup A_{i-1}$ that is in $\al_m$, and let $a_i'$ be a separated arc of $T$ in $\al_m$. Let $F_i'$ be the unique oriented path from $t$ to $a_i'$. Let $c_i$ be the last common vertex of $S$ and $F_i'$, and $a_i$ be the first separated arc of $F_i'(c_i,a_i')$ in $\al_{\ell(a_i')}$. Then $F_i$ is defined as the subpath $F_i'(c_i,a_i)$.
	
	\begin{claim}\label{fuzzy_claim}
		$F_i$ is a fuzzy path for all $0 \leq i \leq q$.
	\end{claim}
	\begin{proof}
		We show first that $F_0$ is fuzzy. Since $a_0 = u_0v_0$ is the first separated arc of $F_0$ in $\al_0$, the last vertex of $F_0$ is $u_0$, and therefore it is in $L_0$. To see this, assume for contradiction that the last vertex of $F_0$ is $v_0$ (in $L_1$). Suppose the arc of $F_0$ before $a_0$ is $u_0w_0$ for some $w_0$ (since $u_0$ is in $L_0$, it has no inneighbours; $u_0w_0$ is not an arc of $S$, because if it was, then $u_0v_0$ would not be separated). Since $a_0$ is the first separated arc of $F_0$, $u_0w_0$ cannot be separated. But that is not possible, because arcs $u_0v_0$ and $u_0w_0$ have the same starting vertex, so they are in the same component of $\al_0$.
		
		Since $u_0 \in L_0$ and $t \in L_n$, $P(u_0,t)$ is fuzzy by Lemma~\ref{fuzzy_if_bt}. Since $F_0$ is a subpath of $P(u_0,t)$, $F_0$ is also a fuzzy path.
		
		Recall that $a_i$ is the first separated arc of $F_i'(c_i,a_i')$ in $\al_{\ell(a_i')}$, so $a_i$ is either in the top or the bottom arc level of $F_i$. Suppose not. Then there are arcs $e \in \al_j$ and $e' \in \al_{j'}$ of $F_i$ such that $j > \ell(a_i')$ and $j' < \ell(a_i')$. Therefore $F_i(e,e')$ contains an arc $f$ in $\al_{\ell(a_i')}$, and clearly, $f$ is separated. This would contradict that $a_i$ is the first arc of $F_i'(c_i,a_i')$ in $\al_{\ell(a_i')}$.
				
		Suppose that $a_i=uv$ is in the bottom arc level of $F_i$. Recall that $a_i$ is the first separated arc of $F_i'(c_i,a_i')$. Assume first that $a_i$ is the only arc of $F_i$ in $\al_{\ell(a_i')}$. Then the last vertex of $F_i$ is $u$, and $u$ is also the only vertex of $F_i$ in the bottom vertex level of $F_i$. Therefore $\ell(c_i) > \ell(u)$ (recall that $c_i$ is the only common vertex of $S$ and $F_i$), and since $S$ is fuzzy, $u$ is in the bottom vertex level also of $P(u,t)$ (but it is possible that $P(u,t)$ contains other vertices in its bottom vertex level). Since vertex $u$ is in the bottom vertex level of $P(u,t)$, and vertex $t$ is in its top level, it follows from Lemma~\ref{fuzzy_if_bt} that $P(u,t)$ is fuzzy. Since $F_i$ is a subpath of $P(u,t)$, $F_i$ must also be fuzzy.
		
		Assume therefore that $a_i$ is not the only arc of $F_i$ in $\al_{\ell(a_i')}$. Let $e$ be the first non-separated arc of $F_i$ in $\al_{\ell(a_i')}$. Such an arc can be only the first arc of $F_i$, and therefore one of the endpoints of $e$ is $c_i$. If $c_i$ is the head of $e$, then since $S$ is fuzzy, all vertices of the subpath $P(c_i,t)$ of $S$ are in a vertex level $L_k$, where $k \geq \ell(u)$. That is $u$ is in the bottom vertex level of $P(u,t)$. So as above, $P(u,t)$ must be fuzzy, and therefore $F_i$ is fuzzy. (However, if $F_i$ is fuzzy and $e$ and $a_i$ are in the same arc level, then $P(e,a_i)$ is a $\mathrm{Z}$, and therefore since $e$ is not separated, $a_i$ cannot be separated either. Therefore such an $F_i$ cannot exist.)
		
		If $c_i$ is the tail of $e$, then let $d$ be the vertex of $P(c_i,t)$ after $c_i$. The arc on vertices $c_i$ and $d$ must be a forward arc $c_id$, since otherwise $e$ would be separated. Since $d \in L_{\ell(u)+1}$ and $S$ is fuzzy, all vertices of $P(c_i,t)$ are in a vertex level $L_k$, where $k \geq \ell(u)$, and we proceed as above.
				
		If $a_i$ is in the top level of $F_i$, then a similar argument works using  $P(s,v)$ instead of $P(u,t)$.
	\end{proof}

	\begin{claim}\label{up_or_down}
		For each $0 \leq i \leq q$, $F_i$ is either a downward or an upward (fuzzy) path.
	\end{claim}
	\begin{proof}
		Assume that $F_i = w_0 w_1 \dots w_n$ (note that $w_0 = c_i$). Suppose for contradiction that $\ell(w_0) = \ell(w_n)$. This implies that $F_i$ must have at least two arcs. Since $F_i$ is fuzzy, it must be that $height(F_i) \leq 2$, because if the height is more than $2$, it is not possible that $\ell(w_0) = \ell(w_n)$. If $height(F_i) = 1$, then the last two arcs of $F_i$ are either $w_{n-1} w_{n-2}$ and $w_{n-1}w_n$, or $w_{n-2} w_{n-1}$ and $w_nw_{n-1}$. In both cases, if the last arc of $F_i$ is separated, then so is the second last arc. This contradicts the definition of $F_i$, namely, that its last arc is the first separated arc in that arc level.
		
		If $height(F_i) = 2$, then it still must be that the last two arcs of $F_i$ are either $w_{n-1} w_{n-2}$ and $w_{n-1}w_n$, or $w_{n-2} w_{n-1}$ and $w_nw_{n-1}$, so we can argue similarly.
	\end{proof}
	
	Let $F_i$ ($0 \leq i \leq q$) be a fuzzy downward (upward) path that contains a separated arc. Let $e$ be the first arc of $F_i$.
	\begin{itemize}
		\item We say that $F_i$ is type $1$ if $e$ is not separated;
		\item We say that $F_i$ is type $2$ if $e$ is separated and $e$ is a backward (forward) arc of $F_i$;
		\item We say that $F_i$ is type $3$ if $e$ is separated and $e$ is a forward (backward) arc of $F_i$.
	\end{itemize}
	
	\begin{claim}\label{no_type_3}
		Let $0 \leq i \leq q$ be an integer such that $F_i$ is downward (upward) and type $3$. Then $T$ contains a fuzzy $\mathrm{N}$.
	\end{claim}
	
		\begin{proof}
			Since the first arc $e$ of $F_i$ is a forward arc and $e$ is separated, $c_i$ cannot have an outneighbours in $S$. (If $S$ has an outneighbor $v_o$ at $c_i$ then, $e$ would be in the same component of $\al_{\ell(e)}$ as $c_i v_0$, contradicting that $e$ is separated.) Since $F_i$ is downward, $F_i$ must have an arc in $\al_{\ell(e) - 1}$. Let $b_1v$ be the first arc of $F_i$ in $\al_{\ell(e) - 1}$. Let $t_1$ be the outneighbor of $c_i$ in $F_i$. Let $b_2$ be the in neighbor of $c_i$ in $S(c_i,t)$, and let $t_2$ be the first vertex of $S(c_i,t)$ in $L_{\ell(t_1)}$. It is easy to check that $b_1,t_1,b_2,t_2$ satisfy the conditions of Lemma~\ref{fuzzy_N_present}, so $T$ contains a fuzzy $\mathrm{N}$. The argument is analogous when $F_i$ is an upward path.
		\end{proof}
	
	\begin{claim}\label{continuous}
			Let $1 \leq i \leq q$. Let $e$ be the first arc of $F_i$. Either all arcs of $F_i$ are separated, or arcs of $F_i$ in $\al_{\ell(e)}$ are non-separated, and all other arcs of $F_i$ are separated.
	\end{claim}
	\begin{proof}
			Assume first that $e$ is separated. If $e'$ is non-separated, then the path $F_i(c_i,e')$ must be a $\mathrm{Z}$, since by definition, $e'$ must be in the same component of $\al_{\ell(e')}$ as some arc of $S$. But since $e$ is the first arc of $F_i$, $F_i(c_i,e')$ contains $e$, so $e$ is also in the same component as $e'$, and therefore $e'$ is also non-separated. This is a contradiction. So all arcs of $F_i$ are separated.
			
			Assume that $e$ is non-separated. Let $e'$ be any arc of $F_i$. If $e' \in \al_{\ell(e)}$, then since $F_i$ is fuzzy, by Claim~\ref{fuzzy_claim}, $F_i(e,e')$ must be a $\mathrm{Z}$, and therefore $e'$ is also non-separated. If $e'$ is not in $\al_{\ell(e)}$, then the path $F_i(c_i,e')$ is not a $\mathrm{Z}$, so $e'$ cannot be connected to an arc of $S$ inside one arc level.
	\end{proof}
		
	By Claims~\ref{up_or_down}~and~\ref{no_type_3}, we can assume that $F_i$ is type $1$ or $2$ for each $0 \leq i \leq q$.
	We prove Claim~\ref{down} now.
		
	\begin{claim}\label{down}
		$F_i$ is a downward (fuzzy) path for all $0 \leq i \leq q$.
	\end{claim}
		\begin{proof}
		 The proof of Claim~\ref{fuzzy_claim} also showed that the last vertex of $F_0$ is in $L_0$. Since by Claim~\ref{up_or_down}, each $F_0$ is either a downward or an upward path, $F_0$ is a downward path.
		
		We show now that the rest of the $F_i$ are also downward. Suppose for contradiction that there is an upward path among $F_0,\dots,F_q$, and let $j$ be the smallest index such that $F_j$ is an upward path. Let $\beta$ be the index such that $F_j$ contains a separated arc in $\al_{\beta+1}$, but not in $\al_\beta$.
		
		\begin{subclaim}
			$\beta$ exists.
		\end{subclaim}
		\begin{proof}
		It is sufficient to show that $F_j$ does not contain a separated arc in $\al_0$. Suppose otherwise. If $\ell(c_j) \geq 2$, then $F_j$ cannot be an upward fuzzy path. If $\ell(c_j) = 0$, then the first arc of $F_j$ is non-separated, so all arcs of $F_j$ in $\al_0$ are non-separated by Claim~\ref{continuous}. If $\ell(c_j) = 1$ then the first arc of $F_j$ must be a backward arc, since otherwise $F_j$ could not be an upward path containing an arc in $\al_0$. If this first arc $f$ is non-separated then we use Claim~\ref{continuous} as above. If $f$ is separated, then $F_j$ is type $3$, and that leads to a contradiction by Claim~\ref{no_type_3}.
		\end{proof}		
		
		 \begin{subclaim}
		 There is an index $j' < j$ such that $F_{j'}$ contains a separated arc $g_{j'}$ in $\al_\beta$.
		 \end{subclaim}
		 \begin{proof}
		  By Claim~\ref{continuous},  if the upward fuzzy path $F_j$ does not contain a separated arc in $\al_{\beta}$ and it contains a separated arc in $\al_{\beta+1}$, then all arcs of $F_j$ must be in arc levels $\al_{\alpha}$ for some $\alpha \geq \beta$. Therefore when $F_j$ is chosen, by the definition of the sequence $F_0,\dots,F_q$, there must be an index $j' < j$ such that $F_{j'}$ contains a separated arc in $\al_{\beta}$. where.		
		\end{proof}
		
		Since $j' < j$, $F_{j'}$ is a downward path.
	
		\begin{subclaim}\label{progressive}
		Let $f$ be a separated arc of $F_i$ for some $0 \leq i \leq q$. Then for each $0 \leq \alpha \leq \ell(f)$, there is an index $i' \leq i$ such that $F_{i'}$ contains a separated arc in $\al_{\alpha}$.
		\end{subclaim}
		\begin{proof}
		This easily follows from the definition of $F_0,\dots,F_q$ and Claim~\ref{continuous}.
		\end{proof}
		
		\begin{subclaim}\label{high_enough}
			Let $f$ be the first arc of $F_{j'}$. Then there exists an $i$ such that $F_j$ contains a separated arc in $\al_i$, where $i >\ell(f)$ if $f$ is separated, and $i \geq \ell(f)$ otherwise.
		\end{subclaim}
		\begin{proof}
			Assume first $f$ is separated. Using Subclaim~\ref{progressive}, for each $\alpha \leq \ell(f)$, there is an index $k$ such that $F_k$ contains a separated arc in $\al_{\alpha}$. Therefore when $F_j$ is defined, it must contain a separated arc in an arc level $\al_i$, where $i > \ell(f)$. If $f$ is non-separated, then since $F_{j'}$ is downward, $F_{j'}$ contains a separated arc in $\al_{\ell(f) - 1}$. Now we can proceed as in the previous case.
		\end{proof}		
	
			\begin{subclaim}
				If both $F_j$ and $F_{j'}$ are type $1$, then $T$ contains a fuzzy $\mathrm{N}$.
			\end{subclaim}
			\begin{proof}
				Since the first arc of $F_j$ and $F_{j'}$ are non-separated and both $F_j$ and $F_{j'}$ contain a separated arc by assumption, both $F_j$ and $F_{j'}$ must have height at least $2$. Let $b_1$ be the first vertex of $F_{j'}$ in $L_\beta$, $t_1$ be a vertex of $F_{j'}$ in the top level of $F_{j'}$. Let $b_2$ be a vertex of $F_{j}$ in level $L_\beta$ (which exists since $F_j$ is type $1$), and $t_2$ be a vertex of $F_{j'}$ in $L_{\ell(t_1)}$ (which exists by Subclaim~\ref{high_enough}). Since both $F_j$ and $F_{j'}$ are type $1$ and $S$ is a fuzzy path, any vertex of $S(t_1,b_2)$ is in level $L_k$ for some $\beta \leq k \leq \ell(t_1)$. It follows that $b_1,t_1,b_2,t_2$ satisfy the conditions of Lemma~\ref{fuzzy_if_bt}, and therefore $T$ contains a fuzzy $\mathrm{N}$.
			\end{proof}
			
			\begin{subclaim}\label{sc22}
				If $F_j$ and $F_{j'}$ are type $2$, then $T$ contains a fuzzy $\mathrm{N}$.
			\end{subclaim}
			\begin{proof}
				The proof is illustrated in Figure~\ref{fig_sc22}. Let $b_1$ be the first vertex of $F_{j'}$ in $L_\beta$. Since $F_{j'}$ is type $2$, $c_{j'}$ cannot have inneighbours in $S$. Let $t_1$ be the outneighbours of $c_{j'}$ that is a vertex of $S(c_{j'}, c_j)$. Since $F_{j'}$ is type $2$, $c_j$ does not have outneighbours in $S$. Let $b_2$ be the inneighbour of $c_j$ in $S(c_{j'},c_j)$. Let $t_2$ be a vertex of $F_j$ in $L_{\ell(t_1)}$. Using Subclaim~\ref{high_enough}, it is easy to check that $b_1,t_1,b_2,t_2$ satisfy the conditions of Lemma~\ref{fuzzy_if_bt}, and therefore $T$ contains a fuzzy $\mathrm{N}$.
			\end{proof}
			
			The cases when one of $F_j$ and $F_{j'}$ is type $1$ and the other is type $2$ can be handled similarly to the previous two claims.
		\end{proof}
		
		\begin{figure}[htb]
			\begin{center}
				\includegraphics[scale=\wp]{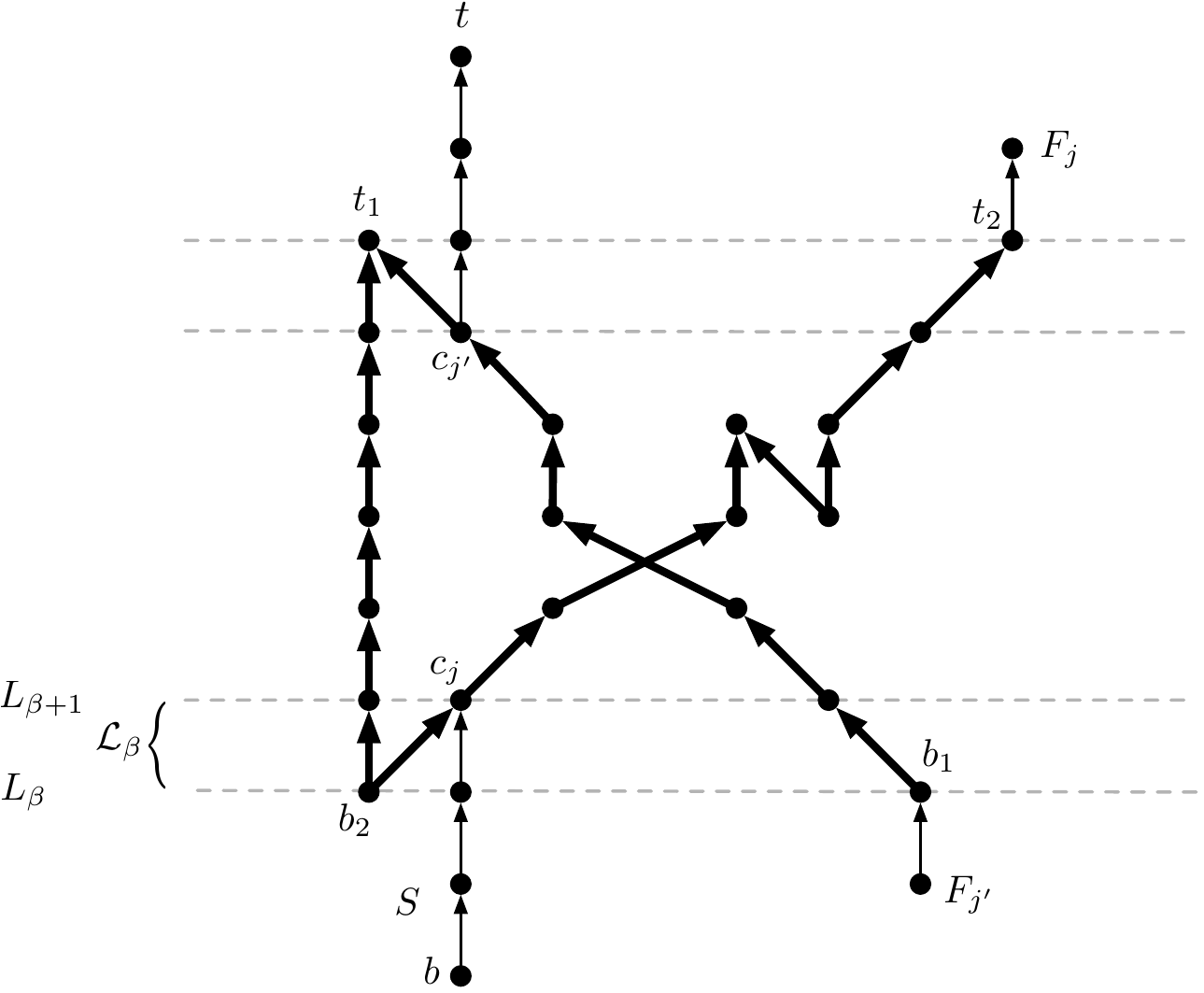}
			\end{center}
			\caption{Illustration of the proof of Subclaim~\ref{sc22}.}\label{fig_sc22}
		\end{figure}

	By Claim~\ref{down}, $F_q$ is a downward path. By definition, $F_q$ contains a separated arc $a_q$ in $\al_{height(T)-1}$. If the path $F_q(c_q,a_q)$ has height at least $2$, then it cannot be downward since $F_q$ is fuzzy. If $F(c_q,a_q)$ has height $1$, then since $a_q$ is separated, all arcs of $F(c_q,a_q)$ must be separated, including the first arc $e$ of $F_q(c_q,a_q)$. Note that $e$ must be a forward arc, since if it is a backward arc, then $c_q \in L_n$, and $S$ must contain an inneighbour of $c_q$, implying that $e$ cannot be separated. It follows that $F(c_q,a_q)$ is type $3$, and the contradiction follows from Claim~\ref{no_type_3}. The lemma is proved.
\end{proof}

Before we can proceed, we need a simple characterization of oriented trees of height $1$ containing no $\mathrm{Z_6}$ as an induced subgraph.

\begin{definition}\label{spider}
	An \emph{in-star} (\emph{out-star}) is an oriented tree with vertex set $\{v_0,v_1,v_2,\dots,v_n\}$, and arc set $\{v_1v_0, v_2v_0,\dots,v_nv_0\}$ ($\{v_0v_1, v_0v_2,\dots,v_0v_n\}$), where $n \geq 0$. Note that if $n = 0$, then the star is just vertex $v_0$. Vertex $v_0$ is called the \emph{root}, and vertices $v_i$ ($1 \leq i \leq n$) are called \emph{leaves}.
	
	Let $S_0$ be an in-star (out-star) with leaves $v_1,\dots, v_n$, and $S_1,\dots,S_n$ be out-stars (in-stars) with root vertices $r_1,\dots,r_n$, respectively. Let $S$ be the oriented tree obtained by
	\begin{itemize}
		\item taking the disjoint union of $S_i$, $0 \leq i \leq n$, and
		\item identifying the leaf vertex $v_j$ of $S_0$ with the root vertex $r_j$ of $S_j$, for each $1 \leq j \leq n$.
	\end{itemize}
	Then $S$ is called an \emph{in-spider (out-spider)}. $S_0$ is called the \emph{body}, and $S_1,\dots,S_n$ are called the \emph{legs} of $S$.
\end{definition}

\begin{lemma}\label{h1}
	Let $T$ be an oriented tree of height $1$. Suppose that $T$ does not contain $\mathrm{Z_6}$ as an induced subgraph. Then $T$ is either an in-spider or an out-spider.
\end{lemma}
\begin{proof}
	Let $L_0$ and $L_1$ be the vertex of levels of $T$. Assume first that $T$ contains an induced $Z_5^{f=0}$ with vertex set $V = \{a,b,c,d,e\}$ and arc set $\{ab,cb,cd,ed\}$. Then $a,c,e \in L_0$ and $b,d \in L_1$. Assume that $L_1 = \{u_1,\dots,u_n\}$. If $cu_i$ is not an arc for some $1 \leq i \leq n$, then since $T$ is connected, there must be an induced oriented path from $u_i$ to a vertex of in $V$. For example, if there is a path from $u_i$ to $e$, then there is an arc $ew$ in $T$, and thus $\{a,b,c,d,e,w\}$ induces a \z6\ in $T$. It is easy to check that $T$ contains an induced \z6\ in the remaining cases.
	
	Let $S_0$ to be the out-star with root $c$ and leaf set $L_1$. For each $1 \leq i \leq n$, let $S_i$ be the in-star with root $u_i$. Let the leaves of $S_i$ be the inneighbours of $u_i$. Clearly,  $T$ is an out-spider with body $S_0$ and legs $S_1,\dots,S_n$.	If $T$ contains a $\zz_5^{f=1}$, then an analogous argument shows that $T$ is an in-spider.
	
	Suppose now that $T$ contains an induced $\zz_4^{f=0}$ with vertex set $\{a,b,c,d\}$ and arc set $\{ab,cb,cd\}$ but not a $\zz_5$. Then $a$ has out-degree $1$ and $d$ has in-degree $1$, since otherwise $T$ would contain a $\zz_5$. Any vertex $w$ of $T$ in $L_1$ can have only $c$ as its inneighbour, and any vertex $z$ in $L_0$ can have only $b$ as its outneighbours. It follows that $T$ is both an in-spider and an out-spider.
	
	The remaining cases are also easy to analyze, e.g., when $T$ contains a $\zz_3$ but not $\zz_4$, then $T$ is up-star or a down-star.
\end{proof}

We are ready to prove one direction of Theorem~\ref{construction_theorem} after the following definition, which will also be used later.
\begin{definition}
	Let $G$ be a leveled digraph and $v$ be a vertex of $G$. The \emph{up-component of $G$ at $v$} is the subgraph of $G$ induced by the set of vertices
	\[
	U = \{u \;|\; \exists \text{ a walk $W$ from $v$ to $u$ such that for each vertex $w$ of $W$,  $\ell(w) \geq \ell(v)$}\}.
	\]
	
	Given a vertex level $L$ of $G$, the \emph{set of up-components of $G$ at level $L$} is the set of digraphs which are up-components for some vertex $v \in L$. A \emph{down-component at $v$} and \emph{set of down-components at level $L$} are defined analogously.
\end{definition}

\begin{lemma}\label{d1}
	If an oriented tree $T$ contains neither a \z6\ nor a fuzzy $\mathrm{N}$ as an induced subgraph, then $T$ is constructible.
\end{lemma}
\begin{proof}
	Let $T$ be as stated in the lemma. We use Lemma~\ref{technical_lemma} to find an integer $\alpha$ such that arc level $\al_\alpha$ of $T$ contains exactly one component $R$. (If there is no such $\alpha$, then $T$ must be a single vertex, and we are done.) By Lemma~\ref{h1}, $R$ is an out-spider or an in-spider. We assume that $R$ is an in-spider. The case when $R$ is an out-spider can be handled similarly. Assume that $S_0$ is the body of $R$ having root $v_0$, and leaves $v_1,\dots,v_n$. We define $T_0$ as the up-component of $T$ at $v_0$ (so $v_0$ is in the bottom vertex level of $T_0$), and $T_1,\dots,T_{n'}$ as the components of $T \setminus T_0$. We show that $T$ is the up-join of $T_0,T_1,\dots,T_{n'}$; this will follow from these claims:
	\begin{enumerate}
		\item The only common vertex of $T_0$ and $R$ is $v_0$, and the only vertex of $T_0$ in $L_{\ell(v_0)}$ that has indegree non-zero (with respect to $T$) is $v_0$;
		\item Each $T_i$ contains precisely one vertex among $v_1,\dots,v_n$. (And by definition, no $T_{i'}$ and $T_{i''}$, $i' \neq i''$ contains the same vertex among $v_1,\dots,v_n$.) We assume w.l.o.g.\ that $v_i$ belongs to $T_i$. It also follows that $n = n'$. Furthermore, any vertex $w \neq v_i$ of $T_i$ in $L_{\ell(v_0) - 1}$ has out-degree zero.
	\end{enumerate}
	Once these are established, it follows that $T$ is the up-join of $T_1,\dots,T_{n'}$. Vertex $v_0$ is in the bottom vertex level of the tree $T_0$. Each $T_i$ contains at most one vertex of degree at least one in $L_{\ell(v_0) - 1}$, which is $v_i$, and $T$ is obtained by taking the disjoint union $T_0,T_1,\dots,T_{n'}$ and adding the arcs of $v_1v_0,\dots,v_nv_0$. We prove the two claims now.

	Suppose for contradiction that $T_0$ contains a vertex $w$ of $R$ such that $w \neq v_0$. Then $w$ is an outneighbours or $v_j$ for some $1 \leq j \leq n$ (because $R$ is an in-spider). Therefore the arcs $v_jw$, $v_jv_0$, and the path from $v_0$ to $w$ in $T_0$ forms a cycle in $T$, a contradiction. It follows that any vertex $w' \neq v_0$ of $T_0$ in $L_{\ell(v_0)}$ has indegree $0$ in $T$, since otherwise $\al_i$ would contain at least two components (contradicting our choice of $i$).
	
	We prove the second claim. Since $T$ is connected, there must be a path from $v_0$ to a vertex of $T_i$. This can only happen if $T_i$ contains at least one of the vertices $v_1,\dots,v_n$. Conversely, suppose for contradiction that there are indices $1 \leq i', i'' \leq n$ such that $v_{i'}$ and $v_{i''}$ belong to the same $T_i$ for some $1 \leq i \leq {n'}$. But then the arcs $v_{i'}v_0$, $v_{i''}v_0$ together with the oriented path in $T_i$ from $v_{i'}$ to $v_{i''}$ would form a cycle, contradicting that $T$ is an oriented tree.
	
	Furthermore, assume for contradiction that there is a vertex $w \neq v_i$ of $T_i$ (for some $1 \leq i \leq n'$) in $L_{\ell(v_0) - 1}$ that has out-degree at least one. Then the arc leaving $w$ is an arc in $\al_\alpha$ that is not part of $R$, a contradiction.
\end{proof}

\begin{proof}[Proof of Theorem~\ref{construction_theorem}.]
	One direction of the theorem follows from Lemma~\ref{d1}. The other direction either follows from the chain of implications outlined at the end of the introduction, or a direct proof is given in Appendix B, see Lemma~\ref{d2}.
\end{proof}

\section{A simple inductive algorithm for LHOM($T$)}\label{algorithm}

We note that if an oriented tree $T$ contains a \z6\ or a \fN\ as an induced subgraph, then LHOM($T$) is $\mathrm{NL}$-hard. This follows from the fact if $T$ contains a \z6\ or a \fN\ as an induced subgraph, then $T$ contains a \cN, as we will show in Theorem~\ref{circ_N_ind}. If $T$ contains a \cN\, then LHOM($T$) is NL-hard by \cite{soda_lhom}.\footnote{However, it is not hard to directly prove that if an oriented tree $T$ contains a \z6\ or a fuzzy \fN\ as an induced subgraph, then LHOM($T$) is $\mathrm{NL}$-hard. Note that we can \emph{primitive-positive define} (\emph{pp-define}) the relation $R = \{(0,0),(0,1),(1,1)\}$ over a \z6\ or a fuzzy $\mathrm{N}$. For example, the pp-definition of $R$ over \z6\ can be done in the same way as it is done for an undirected path on $6$ vertices in \cite{stacs_lhom}. The pp-definition of $R$ over a \fN\ can be done very similarly to the definition of $R$ using a \cN\ (see \cite{soda_lhom}). If $R$ can be pp-defined over $H$, then it is well known that there is a straightforward logspace reduction from the $\mathrm{NL}$-complete directed graph unreachability problem to LHOM($H$).}

In the rest of this section, we inductively construct a logspace algorithm for LHOM($T$). We begin with a high-level description given as Algorithm~\ref{high_level} below. Suppose that $T$ is the up-join of $T_0,T_1,\dots,T_n$, $v_0$ is the central vertex, and $v_1,\dots,v_n$ are the join vertices. (The case when $T$ is the down-join of some trees can be analyzed similarly.) We assume inductively that there is a logspace algorithm $A_i$ for LHOM($T_i$) for each $0 \leq i \leq n$.

\begin{algorithm}
	\caption{High-level algorithm when $T$ is the up-join of $T_0,T_1,\dots,T_n$.}\label{high_level}
	\begin{algorithmic}[1]
		\INPUT A digraph $G$. ($G$ could have many components.)
		\OUTPUT YES if there is a list homomorphism from $G$ to $T$, and NO otherwise.
		\For{each component $G'$ of $G$}
		\State Check whether $G'$ is leveled. If not, output NO.\label{is_leveled}
		\State If the height of $G'$ is larger than the height of $T$, output NO.\label{is_too_high}
		\State Find the height $h$ of $G'$, and find the vertex levels $L_0,L_1,\dots,L_h$ of $G'$.
		\For{each $0 \leq \alpha \leq h - 1$}
		\If{there is a list homomorphism $h$ from $G'$ to $T$ such that $h$ maps level $L_\alpha$ of $G'$ \indent \indent to level $L_{\ell(v_0) - 1}$ of  $T$} mark $G'$ as \textsc{good}\label{main_check}
		\EndIf
		\EndFor
		\If{there is a list homomorphism from $G'$ to $T_i$ for some $0 \leq i \leq n$} mark $G'$ as \textsc{good}\label{to_whole_tree}
		\EndIf
		\EndFor
		\If{all components $G'$ of $G$ are marked \textsc{good}} output YES.
		\Else{ output NO.}\label{really_no}
		\EndIf
	\end{algorithmic}
\end{algorithm}

First we argue the correctness of Algorithm~\ref{high_level}, and then show how to implement line~\ref{main_check}.

\begin{lemma}
	Let $T$ be the up-join of $T_i$, and $A_i$ be the corresponding logspace algorithms for LHOM($T_i$), $0 \leq i \leq n$, as described above. Let $G$ be a digraph.
	Then there is a list homomorphism from $G$ to $T$ if and only if Algorithm~\ref{high_level} on input $G$ outputs YES.
\end{lemma}
\begin{proof}
	Clearly, if Algorithm~\ref{high_level} outputs YES, then there is a list homomorphism from $G$ to $H$. Suppose therefore that Algorithm~\ref{high_level} outputs NO.	If Algorithm~\ref{high_level} outputs NO in line~\ref{is_leveled} or line~\ref{is_too_high}, then clearly, there can be no homomorphism from $G$ to $T$. Assume therefore that Algorithm~\ref{high_level} outputs NO in line~\ref{really_no}. Then at least one of the components $G'$ of $G$ is \emph{not} marked as \textsc{good}. Assume for contradiction that there is a list homomorphism $h$ from $G'$ to $T$. Then either there is a vertex level of $G'$ that is mapped to $L^T_{\ell(v_0)-1}$, or not. There cannot be a vertex level of $G'$ that is mapped to $L^T_{\ell(v_0)-1}$, since then in line~\ref{main_check} the algorithm would have marked $G'$ as good. Therefore $h$ must map $G'$ to $T$ in such a way that no vertex of $G'$ is mapped to a vertex in $L^T_{\ell(v_0)-1}$. That means that $h$ does not map any arc of $G'$ to any of the arcs $v_1v_0,v_2v_0,\dots,v_nv_0$. Since $G'$ is connected, this implies that $h$ must map $G'$ to $T_i$ for some $0 \leq i \leq n$. But that also cannot happen because then $G'$ would be marked \textsc{good} in line~\ref{to_whole_tree}, a contradiction.
\end{proof}

To implement Algorithm~\ref{high_level} in logspace, we need the following lemmas about digraphs.

\begin{lemma}\label{two}
	Let $G$ be a connected acyclic digraph $G$. Fix two arbitrary vertices $u,v \in V(G)$. If $G$ is not leveled, then there are two different walks $W_1$ and $W_2$, both from $u$ to $v$, such that $net(W_1) \neq net(W_2)$
\end{lemma}
\begin{proof}
	Suppose that every walk from $u$ to $v$ has the same net length. For each vertex $w$ of $G$, let $Q_w$ be an arbitrary walk from $u$ to $w$. Let $Q_w'$ be another arbitrary walk from $u$ to $w$ (it could be the same as $Q_w$). If $net(Q_w) \neq net(Q_w')$, then let $W$ be a walk from $w$ to $v$. Then both $Q_w W$ and $Q_w' W$ are walks from $u$ to $v$, and $net(Q_w W) \neq net(Q_w' W)$, a contradiction. Therefore $net(Q_w) = net(Q_w')$.
	
	Assign each vertex $w$ the integer $net(Q_w)$. Let $xy$ an arbitrary arc of $G$, then $net(Q_y) = net(Q_x xy)$, where $Q_x xy$ denotes the walk $Q_x$ and then making one more step from the last vertex $x$ of $Q_x$ to vertex $y$. Therefore $net(Q_y) = net(Q_x) + 1$, i.e., if $xy$ is an arc, then the integer assigned to $y$ is one larger than the integer assigned to $x$. Let $m$ be the minimum of $net(Q_w)$ over $w \in V(G)$. Note that $m \leq 0$ since $net(Q_w) = 0$. Assign each vertex $w$ to level $L_{net(Q_w) - m}$. It follows that $G$ is leveled.
\end{proof}

\begin{lemma}\label{bounded}
	Let $G$ be a connected acyclic digraph $G$. Then $G$ is leveled if and only if every walk in $G$ has bounded net length.
\end{lemma}
\begin{proof}
	If $G$ is leveled then every oriented walk from a vertex $u$ to a vertex $v$ of $G$ has net length $\ell(v) - \ell(u)$, and therefore every walk in $G$ has bounded net length. Conversely, suppose that $G$ is not leveled, and fix two vertices $u$ and $v$ of $G$. By Lemma~\ref{two} there are two different walks $W_1$ and $W_2$ from $u$ to $v$ such that $net(W_1) \neq net(W_2)$. Assume without loss of generality that $net(W_1) > net(W_2)$. Then for any positive integer $k$, $net((W_1 \bar{W}_2)^{k+1}) > k$.
\end{proof}

\begin{definition}
	Given a digraph $G$ and an integer $d$, the \emph{undirected} graph $\cG(G,d)$ is constructed as follows. (For short, we write $\cG$ for $\cG(G,d)$.)
	\begin{itemize}
		\item The vertex set of $\cG$ is defined as $V(\cG) = \{I_j(v) \;|\; v \in V(G) \text{ and } v \in V(G) \text{, and } 0 \leq j \leq d\}$
		\item The edge set of $\cG$ is defined as $E(\cG) = \{(I_j(u), I_{j+1}(v)) \;|\; uv \text{ is a forward arc of $G$, and } 0 \leq j \leq d-1\}$
	\end{itemize}
	We use $I_k$ to denote $\bigcup_{v \in V(G)} I_k(v)$, and $I_k^{\cG}$ to emphasize that $I_k$ is defined with respect to $\cG$.
\end{definition}
Clearly, there is a logspace algorithm that takes $G$ and $d$ as inputs and outputs $\cG(G,d)$. The following lemma is a simple observation.

\begin{lemma}\label{walks}
	Let $G$ be a digraph, and consider $\cG = \cG(G,d)$. There is a walk from a vertex $I_0(u) \in I_0^{\cG}$ to a vertex $I_k(v) \in I_k^{\cG}$ if and only if there is a  walk $W = a_0 \dots a_n$ in $G$ such that $a_0 = u$ and $a_n = v$ such that $net(W) = k$, and for each $0 \leq i \leq n$, $\ell(a_i) \geq \ell(a_0)$ (i.e., $a_0$ is in the bottom level of $W$).
\end{lemma}

\begin{lemma}\label{GG}
	Let $G$ be a connected digraph, and $\cG = \cG(G,d+1)$. Then $G$ is leveled and has height at most $d$ if and only if for each vertex $I_0(u) \in I_0^{\cG}$ and $I_{d+1}(v) \in I_{d+1}^{\cG}$ there is no walk from $I_0(u)$ to $I_{d+1}(v)$ in $\cG(G,d+1)$.
\end{lemma}
\begin{proof}
	Assume that $G$ is leveled and has height at most $d$. Suppose for contradiction that for some vertices $I_0(u) \in I_0^{\cG}$ and $I_{d+1}(v) \in I_{d+1}^{\cG}$ there is a walk $I_0(a_0) \dots I_{d+1}(a_n)$ in $\cG(G,d+1)$. Then the walk $a_0\dots a_n$ has net length $d+1$, so $G$ cannot have height at most $d$.
	
	Conversely, assume that $G$ is not leveled or has height at least $d+1$. In both cases (in the former case using Lemma~\ref{bounded}), there is a walk $W = a_0 \dots a_n$ in $G$ such that $a_0 = u$ and $a_n = v$ such that $net(W) = d+1$, and for each $0 \leq i \leq n$, $\ell(a_i) \geq \ell(a_0)$ (i.e., $a_0$ is in the bottom level of $W$). Therefore by Lemma~\ref{walks}, there is a walk from $I_0(u)$ to $I_{d+1}(v)$ in $\cG(G,d+1)$.
\end{proof}

\begin{lemma}\label{component_in_L}
	Let $G$ be a leveled digraph of height at most $h$ (a constant), and $v \in V(G)$. Then there is a logspace algorithm that outputs the up-component (down-component) of $G$ at $v$.
\end{lemma}
\begin{proof}
	We produce $\cG = \cG(G,h)$. We output every vertex $u$ such that there is an undirected path (using Reingold's algorithm \cite{reingold}) from $I_0(v)$ to $I_j(u)$ in $\cG$ for some $j$. These vertices form $U$. Now we output every arc $ab$ such that $a,b \in U$. The down-component at $v$ can be produced in a similar way.
\end{proof}

We are ready to specify the main subroutine (Algorithm~\ref{level_alg}) of Algorithm~\ref{high_level}. We denote the disjoint union of the trees $T_1,\dots,T_n$ with $T'$. Inductively, we assume that there is a logspace algorithm $A_i$ for each of LHOM($T_i$), $0 \leq i \leq n$. We can easily combine the algorithms for LHOM($T_i$), $1 \leq i \leq n$, to obtain a logspace algorithm $A'$ for LHOM($T'$) (we use Reingold's logspace algorithm for undirected reachability \cite{reingold} to output the components $G'$ of $G$, and then test whether there is a list homomorphism from each $G'$ to one of the $T_i$-s.)
\begin{algorithm}
	\caption{Check if there is a list homomorphism $h : G \rightarrow T$ such that $h(L_\alpha^G) \subseteq L_{\ell(v_0)-1}^T$.}\label{level_alg}
	\begin{algorithmic}[1]
		\INPUT A leveled digraph $G$, and an integer $0 \leq \alpha \leq height(G) - 1$.
		\OUTPUT YES if there is a list homomorphism $h : G \rightarrow T$ such that $h(L_\alpha^G) \subseteq L_{\ell(v_0)-1}^T$ and NO otherwise.
		\State Let $\cU_{all}$ be the set of up-components of $G$ at level $L_{\alpha+1}^G$.
		\State Using $A_0$, check for each $U \in \cU$ if there is a list homomorphism $h$ from $U$ to $T_0$ such that for each $v \in V(U) \cap L_{\alpha+1}^G$ that has at least one inneighbour when considered as a vertex of $G$, $h(v) = v_0$ (this can be enforced by setting the list of $v$ to $\{v_0\}$). Let $\cU$ be the set of those $U \in \cU_{all}$ for which such a list homomorphism exists. \label{top_dominant}
		\State Let $G'$ be the subgraph of $G$ induced by the vertices $V(G) \setminus V(\cU)$, where a vertex $v$ belongs to $V(\cU)$ if it is a vertex of some up-component in $\cU$.
		\State Using $A'$, check if there is a list homomorphism $h$ from $G'$ to $T'$ such that vertices of $G'$ in level $L_\alpha^G$ are mapped to vertices of $T'$ in $L_{\ell(v_0) -1}^T$. If no such list homomorphism exists, output NO.\label{reject}
		\State Otherwise output YES.
	\end{algorithmic}
\end{algorithm}

The following lemma proves the correctness of Algorithm~\ref{level_alg}.
\begin{lemma}
	Suppose that $T$ is the up-join of $T_0,T_1,\dots,T_n$, $v_0$ is the central vertex of $T_0$, and $v_i$ is the join vertex of $T_i$, $1 \leq i \leq n$. Let $G$ be a leveled digraph, and $0 \leq \alpha \leq height(G) - 1$ be an integer.  Then there is a list homomorphism $h : G \rightarrow T$ such that level $L_\alpha$ of $G$ is mapped to level $L_{\ell(v_0)-1}$ of $T$, (i.e., $h(L_\alpha^G) \subseteq L_{\ell(v_0)-1}^T$) if and only if Algorithm~\ref{level_alg} outputs YES on input $G,\alpha$.
\end{lemma}
\begin{proof}
	Assume that Algorithm~\ref{level_alg} outputs YES. (The proof is aided by Figure~\ref{alg_2}.) Then there is a list homomorphism $h'$ from $G'$ to $T'$, and a list homomorphism $h_U$ from each up-component $U \in \cU$ to $T_0$. Since all arcs $v_1v_0,v_2v_0,\dots,v_nv_0$ are present, the map that $h(v)$ defined as
	$h'(v)$ if $v \in V(G')$, and as $h_U(v)$ if $v \in V(U)$ is a list homomorphism from $G$ to $T$.
	
	Conversely, assume that there is a list homomorphism $g$ from $G$ to $T$. If there is a list homomorphism from an up-component $U$  to $T_0$, then we can assume that $g$ maps $U$ to $T_0$. After these up-components are removed from $G$ to obtain $G'$, $g|_{G'}$ is a list homomorphism from $G'$ to $T'$. Therefore the algorithm accepts.
	
	\begin{figure}[htb]
		\begin{center}
			\includegraphics[scale=\wp]{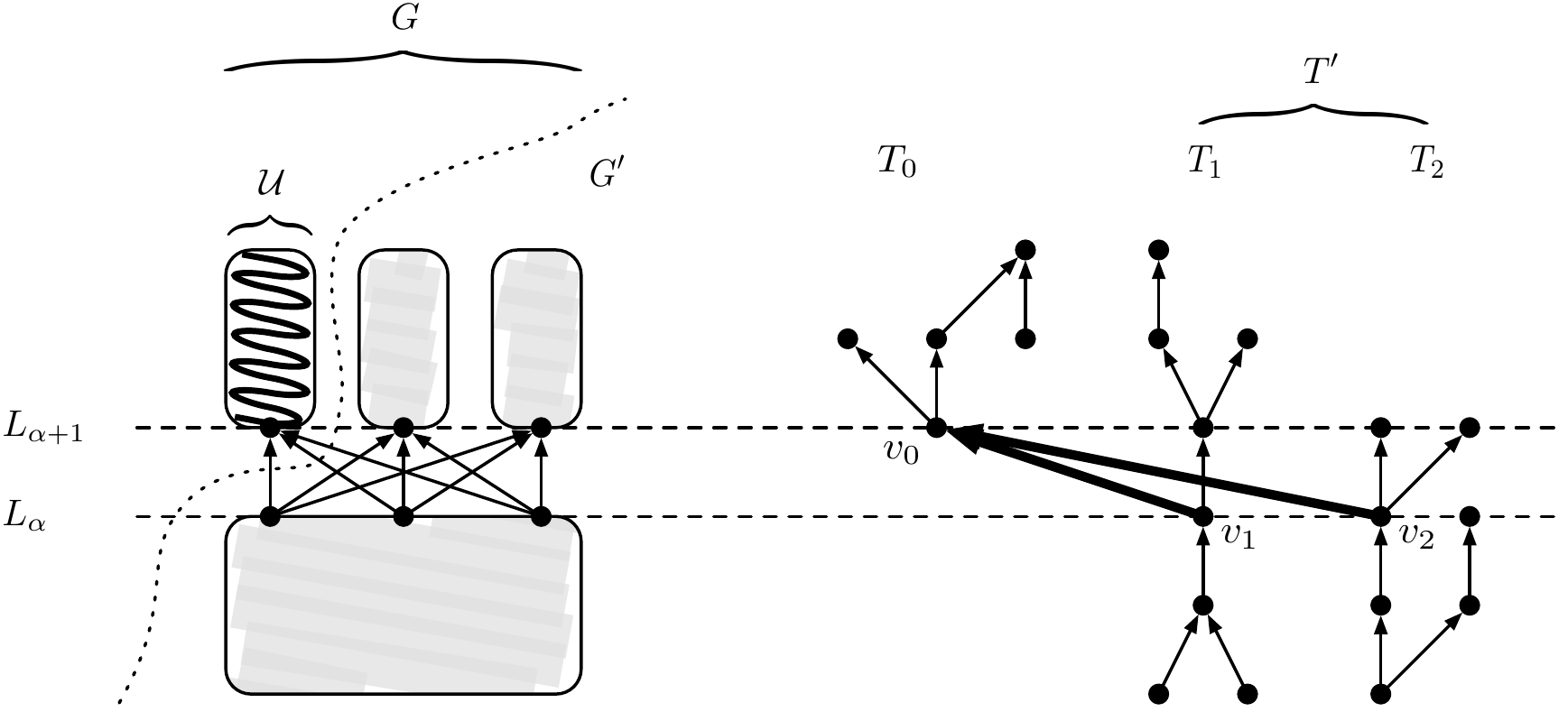}
		\end{center}
		\caption{Illustration of the correctness proof of Algorithm~\ref{level_alg}.}\label{alg_2}
	\end{figure}
\end{proof}

It is routine to implement Algorithms~\ref{high_level} and \ref{level_alg} so that they use only logarithmic space. This is done using the basic trick that if $A_1$ and $A_2$ are logspace algorithms, and $A_3$ is the algorithm that first runs $A_1$ on the input, then feeds the output of $A_1$ to $A_2$, and then outputs the output of $A_2$, then $A_3$ can be assumed to be a logspace algorithm. The various reachability tests the algorithms use can be implemented using Reingold's logspace algorithm for undirected reachability \cite{reingold}.

\section{Circular N, $\mathrm{Z_6}$, and fuzzy N}\label{equivalence}

We prove that if $H$ is a digraph, then $H$ contains a fuzzy $\mathrm{N}$ or $\mathrm{Z_6}$ as an induced subgraph if and only if $H$ contains a circular $\mathrm{N}$. Note that assuming that $L \neq NL$, there is a simpler proof using already proved results.\footnote{Assume for contradiction that $T$ contains a circular \cN\ but no $\mathrm{Z_6}$ or fuzzy $\mathrm{N}$ as an induced subgraph. If $T$ contains no $\mathrm{Z_6}$ or fuzzy $\mathrm{N}$, then we have a logspace algorithm for LHOM($T$) by the previous results of this paper. Since $T$ contains a circular $\mathrm{N}$, LHOM($H$) is $\mathrm{NL}$-hard (\cite{soda_lhom}), so our logspace algorithm works for an $\mathrm{NL}$-hard problem, and therefore $NL = L$, a contradiction.} However, we wish to prove this without the assumption that $L \neq NL$, as stated in the following theorem.

\begin{theorem}\label{circ_N_ind}
	An oriented tree $T$ contains a $\mathrm{Z_6}$ or a fuzzy $\mathrm{N}$ as induced subgraphs if and only if $T$ contains a circular $\mathrm{N}$.
\end{theorem}

We begin with recalling some definitions from \cite{soda_lhom}.
Let $H$ be a digraph. We define two walks $X = x_0 x_1 \dots x_n$ and $Y = y_0 y_1 \dots y_n$ in $H$ to be {\em congruent} if they follow the same pattern of forward and backward arcs, i.e., $x_ix_{i+1}$ is a forward arc if and only if $y_iy_{i+1}$ is a forward arc. Suppose $X, Y$ and $Z = z_0 z_1 \dots z_n$ are congruent walks. We say that $x_iy_{i+1}$ is a {\em faithful arc from $X$ to $Y$} if it is an arc of $H$ in the same direction (forward or backward) as $x_ix_{i+1}$. We say that $X$ {\em avoids} $Y$ in $H$ if there is no faithful arc from $X$ to $Y$ in $H$. Observe that two walks of length zero also avoid each other.

We say that $Z$ {\em protects} $Y$ from $X$ if the existence of faithful arcs $x_iz_{i+1}$ and $z_jy_{j+1}$ in $H$ implies that $j \leq i$. In other words, $Z$ protects $Y$ from $X$ if and only if there exists a subscript $s$ such that $x_0, x_1, \dots, x_s$ avoids $z_0, z_1, \dots, z_s$ and $z_{s+1}, z_{s+1}, \dots, z_n$ avoids $y_{s+1}, y_{s+2}, \dots, y_n$.

\begin{definition}\label{def-cN}
	Let $x, x', y, y'$ be vertices of a digraph $H$. An {\em extended $N$ from $x,x'$ to $y,y'$} in $H$ consists of congruent walks $X$ (from $x$ to $x'$), $Y$ (from $y$ to $y'$), and $Z$ (from $y$ to $x'$), such that $X$ avoids $Y$ and $Z$ protects $Y$ from $X$. A {\em circular $\mathrm{N}$} is an extended $N$ in which $x = x'$ and $y = y'$.
\end{definition}

\begin{lemma}\label{z->cN}
	If an oriented tree $T$ contains a $\mathrm{Z_6}$ or a fuzzy $\mathrm{N}$ as induced subgraphs, then $T$ contains a circular $\mathrm{N}$.
\end{lemma}
\begin{proof}
	We first show how to define $X$,$Y$ and $Z$ when $G$ is an induced a $\mathrm{Z_6}$ with vertex set $\{0,1,2,3,4,5\}$ and arcs $\{01,21,23,43,45\}$.  Set $X = 01010$, $Y =45454$, and $Z = 43210$. It is trivial to check that $X$ avoids $Y$ and $Z$ protects $Y$ from $X$.
	
	Assume that the induced subgraph $H$ is a fuzzy $\mathrm{N}$ that can be expressed as $P_1T\bar{P}_2DP_3$. Suppose that the first vertices of $P_1$ and $P_3$ are $x_0$ and $y_0$, respectively, and the last vertices of $P_1$ and $P_3$ are $x_0'$ and $y_0'$, respectively. Let $h = height(P_1)$.
	
	We define an oriented path $Q$ as follows. Let $Q' = Q_1' Q_2' \dots Q_h'$, where each $Q_i'$ is a $Z_4^{f=0}$. Then $Q$ is define as $Q = Q' \bar{Q'}$. Let $q$ denote the first and $q^*$ denote the last vertex of $Q$. Notice that because $P_1$, $P_2$ and $P_3$ are minimal fuzzy paths and $T$ and $B$ are $Z_3^{f=1}$ and $Z_3^{f=0}$, there exist three homomorphisms $h_1$, $h_2$ and $h_3$ as follows.
	\begin{itemize}
		\item $h_1$ maps $Q$ to $P_1$ such that $h_1(q) = h_1(q^*) =x_0$.
		\item $h_3$ maps $Q$ to $P_3$ such that $h_3(q) = h_3(q^*) =y_0$.
		\item Recall that $Q=Q'\bar{Q'}$. Homomorphism $h_2$ maps $Q'$ to $P_2$ and $\bar{Q'}$ to $\bar{P}_1$ such that $h_2(q) = y_0$ and $h_2(q^*) =x_0$.
	\end{itemize}
	Suppose that $Q = q_0q_1 \dots q_n$, and set $X = h_1(q_0) \dots h_1(q_n)$, $Y = h_3(q_0) \dots h_3(q_n)$, and $Z = h_2(q_0) \dots h_2(q_n)$. Since $X$ contains only vertices of $P_1$ and $Y$ contains only vertices of $P_3$, there is no arc a vertex of $X$ and a vertex of $Y$, and therefore $X$ avoids $Y$. It is also straightforward to verify that $Z$ protects $Y$ from $X$.
\end{proof}

\begin{proof}[Proof of Theorem~\ref{circ_N_ind}]
		One direction of the theorem follows from Lemma~\ref{z->cN}. The other direction either follows from the chain of implications outlined at the end of the introduction, or a direct proof is given in Appendix B, see Lemma~\ref{cN->z}.
\end{proof}

\section{An algebraic characterization}\label{algebra_sec}

In this section, we prove that an oriented tree that does not contain a \z6\ or a \fN\ as an induced subgraph admits a chain of Hagemann-Mitschke polymorphisms of length $3$ (Theorem~\ref{HM_3}). It is quite easy to see that this is not the case for general digraphs. For the sake of completeness, we give an explicit digraph that enjoys a Hagemann-Mitschke chain of conservative polymorphisms of length $n$ but not length $n-1$. This is the content of Theorem~\ref{ladder_thm}.

\subsection{Hagemann-Mitschke chain}

\begin{theorem}\label{HM_3}
	Let $T$ be an oriented tree. Then $T$ has conservative polymorphisms $f_1$, $f_2$ and $f_3$ that form a HM-chain if and only if $T$ does not contain a $\mathrm{Z_6}$ or a fuzzy $\mathrm{N}$ as an induced subgraph.
\end{theorem}

The following lemma proves one direction of the theorem. (The basic idea of the proof of this lemma is inspired by the proof of Lemma~18 in \cite{stacs_lhom}.)
\begin{lemma}\label{HM_chain_defined}
	Let $T$ be an oriented tree. If $T$ can be constructed using Definition~\ref{constructible}, then $T$ has conservative polymorphisms $f_1$, $f_2$ and $f_3$ that form a HM-chain.
\end{lemma}
\begin{proof}
	We show the existence of the claimed conservative polymorphisms using induction on the construction of $T$ given in Definition~\ref{constructible}. The defined operations will be trivially conservative, and we won't mentioned this explicitly. We will work with the up-join operation, and note that the proof works similarly for the down-join operation.
	
	We begin with breaking down the construction in Definition~\ref{constructible} into two steps. Assume that $T$ is the up-join of $T_0,T_1,\dots,T_n$ with central vertex $v_0$ and join vertices $v_1,\dots,v_n$. Taking the up-join can be thought of as taking the disjoint union $T' = T_1 \sqcup \dots \sqcup T_n$ of $T_1,\dots,T_n$, and then taking the disjoint union $T' \sqcup T_0$ and adding the arcs $v_iv_0$, $1 \leq i \leq n$.
		
	First we show that if each of $T_1,\dots,T_n$ admits a HM-chain of polymorphisms of length $3$, then so does $T_1 \sqcup \dots \sqcup T_n$.	Let $f_1^i, f_2^i, f_3^i$ be the desired polymorphisms for $T_i$, $1 \leq i \leq n$.	If $(x,y,z) \in V(T_j)^3$ for some $1 \leq j \leq n$, then for each $1 \leq s \leq 3$, let $g_s(x,y,z) = f_s^j(x,y,z)$. If $(x,y,z) \in V(T_k) \times V(T_l) \times V(T_m)$ such that $|\{k,l,m\}| > 1$, then let $g_1(x,y,z) = x$ and $g_3(x,y,z) = z$, and furthermore, $g_2(x,y,z) = z$ if $k = l$ and $g_2(x,y,z) = x$ otherwise. It is easy to check that $g_1,g_2,g_3$ form a HM-chain, and that each $g_s$ is a conservative polymorphism of $T'$.
		
	Suppose now that $f_1, f_2, f_3$ and $g_1,g_2,g_3$ are the desired polymorphisms for $T_0$ and $T' = T_1 \sqcup \dots \sqcup T_n$, respectively, and that $T$ is obtained by adding arcs $v_1v_0,\dots,v_nv_0$ to $T_0 \sqcup T'$. Let $m = height(T)$. We argue first that it is sufficient to define polymorphisms for vertices $x,y,z$ that are all in the same vertex level $L_j^T$, where $0 \leq j \leq m$. For suppose that $F_1',F_2',F_3' : L_0^3 \cup \dots \cup L_m^3 \rightarrow T$ are operations that are edge-preserving, conservative, and satisfy all required identities. Then we can extend these to full operations with the same properties:
	\begin{align*}
	F_1(x,y,z)  &= \begin{cases}
	F_1'(x,y,z) &\text{ if $x,y,z \in L_j^T$ for some $0 \leq j \leq m$,}\\
	x &\text{ otherwise.}
	\end{cases}\\
	F_3(x,y,z) &= \begin{cases}
	F_3'(x,y,z) &\text{ if $x,y,z \in L_j^T$ for some $0 \leq j \leq m$,}\\
	z &\text{ otherwise.}
	\end{cases}\\
	F_2(x,y,z) &= \begin{cases}
	F_2'(x,y,z) &\text{ if $x,y,z \in L_j^T$ for some $0 \leq j \leq m$, else}\\
	z &\text{ if $x,y \in L_j^T$ for some $0 \leq j \leq m$,}\\
	x &\text{ otherwise.}
	\end{cases}\\
	\end{align*}
	The claimed properties are easy to check. To see that the defined operations are arc-preserving, note that if $xx'$,$yy'$ and $zz'$ are arcs and $x,y,z$ are in levels $L_{i(x)}, L_{i(y)},L_{i(z)}$, respectively, then $x',y',z'$ are in levels $L_{i(x)+1}, L_{i(y)+1},L_{i(z)+1}$, respectively, so the same case in the above definition applies for both $(x,y,z)$ and $(x',y',z')$.
		
	Assume therefore that $T$ is obtained by adding arcs $v_iv_0$, $1 \leq i \leq n$ to $T_0 \sqcup T'$. Using the induction hypothesis and the above argument for disjoint union, we can assume that $T_0$ admits the desired operations $f_1,f_2,f_3$, and $T'$ admits the desired operations $g_1,g_2,g_3$. First we define $F_1$ and $F_3$ for $T$ as follows.
	
	\[
	F_1(x,y,z) = \begin{cases}
	f_1(x,y,z) &\text{ if $x,y,z \in V(T_0)$, else}\\
	g_1(x,y,z) &\text{ if $x,y,z \in V(T')$, else}\\
	x &\text{ if $x \in V(T_0)$ or $y,z \in V(T_0)$, else}\\
	u &\text{ where $u$ is leftmost of $\{y,z\} \cap V(T_0)$.}
	\end{cases}
	\]
	
	\[
	F_3(x,y,z) = \begin{cases}
	f_3(x,y,z) &\text{ if $x,y,z \in V(T_0)$, else}\\
	g_3(x,y,z) &\text{ if $x,y,z \in V(T')$, else}\\
	z &\text{ if $z \in V(T_0)$ or $x,y \in V(T_0)$, else}\\
	u &\text{ where $u$ is leftmost of $\{x,y\} \cap V(T_0)$.}
	\end{cases}
	\]
	
	Notice that $F_1$ is well defined, since if the first three cases do not apply, then in the last case, (precisely) one of $y$ and $z$ is in $V(T_0)$. We can argue similarly for $F_3$.	
	
	Observe that $F_1(x,y,y) = x$. In the first two cases, this follows from the induction hypothesis.  Otherwise, the third line of the definition sets the value of $F_1(x,y,y)$ to $x$. (The last case cannot occur.) Similarly, we can check that $F_3(x,x,y) = y$.
	
	We verify that $F_1$ is arc-preserving, and note that $F_3$ can be analyzed similarly. Assume that $xx'$, $yy'$, and $zz'$ are arcs of $T$.
	\begin{itemize}
		\item If $\ell(x) < \ell(v_0) - 1$ or $\ell(x) > \ell(v_0) - 1$, then we observe that for each $w \in \{x,y,z\}$, it holds that $w \in V(T_0) \Leftrightarrow w' \in V(T_0)$. Also recall that $V(T_0)$ and $V(T')$ partition $V(T)$. Therefore the same case of the above definition applies for both $F_1(x,y,z)$ and $F_1(x',y',z')$. It follows from this and the induction hypothesis that $F_1(x,y,z)F_1(x',y',z')$ is an arc of $T$.
		\item Assume therefore that $\ell(x) = \ell(v_0) - 1$. If $x',y',z' \in V(T')$, then again, $x,x',y,y',z,z' \in V(T')$ and we are done. We note that if any of $x',y',z'$ is a vertex in $V(T_0)$, that vertex must be $v_0$, since $v_0$ is the only vertex of $T_0$ in $L_{\ell(v_0)}^T$ that has an inneighbour in $T$. We also note that if $F_1(x',y',z') = v_0$, then since $F_1(x,y,z) \in \{v_1,\dots,v_n\}$ and $v_iv_0$ are arcs for all $1 \leq i \leq n$, we are done.
		
		If $x',y',z' \in V(T_0)$, then $x'=y'=z'=v_0$, so $F_1(x',y',z') = v_0$. If $x' \in V(T_0)$, then line $3$ of the definition gives that $F_1(x',y',z') = v_0$.
		
		Assume therefore that $x' \in V(T')$ and $y', z' \in V(T_0)$. Then we have by definition that $F_1(x',y',z') = x'$. If we show that $F_1(x,y,z) = x$, then since $xx'$ is an arc, we are done. Recall that $x,y,z \in \{v_1,\dots,v_n\}$, so $x,y,z \in V(T')$, and thus $F_1(x,y,z) = g_1(x,y,z)$. By the definition of $g_1$ above (recall that $g_1$ is over the disjoint union $T' = T_1 \sqcup \dots \sqcup T_n$), if not all of $x,y,z$ are in the same component of $T'$, then $g_1(x,y,z) = x$. Therefore $F_1(x,y,z) = x$. So we can assume that $x,y,z$ are all in the same component $T_i$ of $T'$, for some $1 \leq i \leq n$. Then since the only arc from $T_i$ to $v_0$ is $v_iv_0$ and $y'=z'=v_0$, we have that $y=z=v_i$. By the induction hypothesis $g_1(x,y,y) = x$, so $F_1(x,y,z) = x$.
		
		If $y' \in V(T_0)$ and  $z' \in V(T')$, or if $y' \in V(T')$ and $z' \in V(T_0)$, then  line $4$ of the definition sets $F_1(x',y',z') = v_0$.
	\end{itemize}
	
	We define
	\[
	F_2(x,y,z) = \begin{cases}
	F_1(x,x,z) &\text{ if $x \in V(T')$ and $y,z \in V(T_0)$, or if $x \in V(T_0)$ and $y,z \in V(T')$, else}\\
	F_3(x,z,z) &\text{ if $x,y \in V(T_0)$ and $z \in V(T')$, or if $x,y \in V(T')$ and $z \in V(T_0)$, else}\\
	f_2(x,y,z) &\text{ if $x,y,z \in V(T_0)$, else}\\
	g_2(x,y,z) &\text{ if $x,y,z \in V(T')$, else}\\
	w, &\text{ where $w$ is leftmost of $\{x,y,z\} \cap V(T_0)$.}
	\end{cases}
	\]
	Note that $F_2$ is well defined. In particular, in line $5$, at least one of $x,y,z$ must be in $V(T_0)$, since otherwise line $4$ applies.
	
	To complete the proof that $F_1$, $F_2$, and $F_3$ form a HM-chain, we show that $F_1(x,x,z) = F_2(x,z,z)$ and $F_2(x,x,z) = F_3(x,z,z)$. We focus on $F_1(x,x,z) = F_2(x,z,z)$, and note that it can be shown similarly that $F_2(x,x,z) = F_3(x,z,z)$. If line $1$ of the definition of $F_2$ applies, then we are done by definition. Note that line $2$ cannot not apply for $F_2(x,z,z)$. If line $3$ applies, then $x = z = v_0$, so $F_2(x,z,z) = v_0 = F_1(x,x,z)$. If line $4$ applies, then $F_2(x,z,z) = g_2(x,z,z) = g_1(x,x,z) = F_1(x,x,z)$, where the second equality is by the induction hypothesis, and the last equality is by the definition of $F_1$. Line $5$ cannot apply for $F_2(x,z,z)$.
	
	It remains to show that $F_2$ is arc preserving.
	\begin{itemize}
		\item Suppose that $\ell(x) < \ell(v_0) - 1$ or $\ell(x) > \ell(v_0) - 1$. As before, for each $w \in \{x,y,z\}$, $w \in V(T_0) \Leftrightarrow w' \in V(T_0)$. Therefore the same case of the above definition applies for both $F_2(x,y,z)$ and $F_2(x',y',z')$, and thus $F_2(x,y,z)F_2(x',y',z')$ is an arc of $T$.
		\item Suppose that $\ell(x) = \ell(v_0) - 1$. As before, we use the fact if $F_2(x',y',z') = v_0$, we are done.
		\begin{itemize}
			\item If \emph{line $1$} of the definition of $F_2$ applies for $F_2(x',y',z')$, then consider first when $x' \in V(T')$ and $y',z' \in V(T_0)$. Then $F_2(x',y',z') = F_1(x',x',v_0) = v_0$ (since $z' = v_0)$) by line $4$ of the definition of $F_1$.
			
			Consider therefore the case when $x' \in V(T_0)$ and $y',z' \in V(T')$. Then $F_2(x',y',z') = F_1(v_0,v_0,z') = v_0$ by line $3$ of the definition of $F_1$.
			\item If \emph{line $2$} of the definition of $F_2$ applies, then we can do a similar analysis as above.			
			\item If \emph{line $3$} applies for $F_2(x',y',z')$, then $F_2(x',y',z') = f_2(v_0,v_0,v_0)= v_0$, since $f_2$ is conservative.
			\item If \emph{line $4$} applies for $F_2(x',y',z')$. Then line $4$ applies for $F_2(x,y,z)$, so we are done by the induction hypothesis for $g_2$.
			\item If \emph{line $5$} applies, then $F_2(x',y',z') = v_0$.
		\end{itemize}
	\end{itemize}
\end{proof}

\begin{proof}[Proof of Theorem~\ref{HM_3}]
	One direction is shown in Lemma \ref{HM_chain_defined}. The other direction follows from the chain of implications outlined in the introduction, or we can use Lemma~\ref{no_HM} in the Appendix B.
\end{proof}

\subsection{A digraph with an HM-chain of length n but not n-1}

\begin{definition}\label{ladder}
	A \emph{ladder} of height $n$ is a digraph having the following arcs: $a_0a_1, a_1a_2,\dots,a_{n-1}a_n$, $b_0b_1,b_1b_2,\dots,b_{n-1}b_n$, and $b_0a_1,b_1a_2,\dots,b_{n-1}a_n$.
\end{definition}

\begin{theorem}\label{ladder_thm}
	Let $H$ be a ladder of height $n \geq 1$. Then $H$ admits an HM-chain of conservative polymorphisms of length $n+1$, but not of length $n$.
\end{theorem}
\begin{proof}
The proof follows from Lemmas~\ref{no_HM_ladder} and \ref{ladder_pol} below.
\end{proof}

\begin{lemma}\label{no_HM_ladder}
	Let $H$ be a ladder of height $n \geq 1$. Then $H$ does not admit an HM-chain of conservative polymorphisms $f_1,f_2,\dots,f_n$ of length $n$.
\end{lemma}
\begin{proof}
	For ladder $H$ we use the same notation as in Definition~\ref{ladder}. Suppose for contradiction that $f_1,\dots,f_n$ is an HM-chain of conservative polymorphisms of $H$. We show by induction that $f_i(a_i,a_i,b_i) = a_i$ for $1 \leq i \leq n$, and this will contradict the definition of an HM-chain of length $n$ requiring that $f_n(a_n,a_n,b_n) = b_n$. Since $a_0a_1,b_0a_1,b_0b_1$ are arcs, $f_1(a_0,b_0,b_0)f_1(a_1,a_1,b_1)$ is an arc of $H$. Since $a_0 = f_1(a_0,b_0,b_0)$ (by definition), it follows that $f_1(a_1,a_1,b_1) = a_1$, so the base case holds. Assume the induction hypothesis holds for index $i$. Then $a_i = f_i(a_i,a_i,b_i) = f_{i+1}(a_i,b_i,b_i)$, and since $f_{i+1}(a_i,b_i,b_i)f_{i+1}(a_{i+1},a_{i+1},b_{i+1})$ is an arc of $H$, this arc can only be $a_ia_{i+1}$, so $f_{i+1}(a_{i+1},a_{i+1},b_{i+1}) = a_{i+1}$, and we are done.
\end{proof}

\begin{lemma}\label{ladder_pol}
	Let $H$ be a ladder of height $n \geq 1$. Then $H$ admits an HM-chain of conservative polymorphisms $f_1,f_2,\dots,f_{n+1}$ of length $n+1$.
\end{lemma}
\begin{proof}
	For ladder $H$ we use the same notation as in Definition~\ref{ladder}. Notice that $H$ is leveled. We argue now that it is sufficient to define polymorphisms for vertices $x,y,z$ that are all in the same vertex level $L_j$ of $T$, where $0 \leq j \leq n$. For suppose that $F_1',F_2',F_3' : L_0^3 \cup \dots \cup L_m^3 \rightarrow T$ are operations that are edge-preserving, conservative, and satisfy all required identities. Then we can extend these to full operations with the same properties:
	\[F_1(x,y,z) = \begin{cases}
	F_1'(x,y,z) &\text{ if $x,y,z \in L_j$ for some $0 \leq j \leq m$, else}\\
	x &\text{ if $y,z \in L_j$ for some $0 \leq j \leq m$,}\\
	z &\text{otherwise}
	\end{cases}\]
	and for $2 \leq i \leq n+1$,
	\[F_i(x,y,z) = \begin{cases}
	F_i'(x,y,z) &\text{ if $x,y,z \in L_j$ for some $0 \leq j \leq m$,}\\
	z &\text{ otherwise.}
	\end{cases}\]
	The claimed properties of this extension are easy to check.
	
	So let $x,y,z \in L_j$ for some $0 \leq j \leq n$. We define the operations $f_i(x,y,z)$ for each $1 \leq i \leq n+1$ as shown below. (This definition is inspired by the definition of an HM-chain in Lemma~5.2 in \cite{soda_lhom}.)  Recall that $f_i$ is conservative, so $f_i(x,x,x)$ is always required to be $x$. Also note that we won't discuss the cases in the proofs below which involve $f_i(x,x,x)$, since these cases are trivial to analyze.
	\begin{center}
		\begin{minipage}{0.75\textwidth}
			\begin{center}
				\begin{eqnarray}
					f_i(a_j,a_j,b_j)  &= a_j  & \text{if $j+1 > i$} \\
					&= b_j & \text{if $j+1 \leq i$}\\
					f_i(a_j,b_j,b_j)  &= b_j  & \text{if $j+1 < i$} \\
					&= a_j &\text{if $j+1 \geq i$}\\
					f_i(b_j,a_j,b_j)  &= a_j &\text{if $j+1 > i > n-j+1$}\\
					&= b_j & \text{otherwise}\\
					f_i(b_j,b_j,a_j)  &= b_j & \text{if $n - j+1 > i$}\\
					&= a_j & \text{if $n - j+1 \leq i$}\\
					f_i(b_j,a_j,a_j)  &= a_j & \text{if $n - j+1 < i$}\\
					&= b_j & \text{if $n - j+1 \geq i$}\\
					f_i(a_j,b_j,a_j)  &= b_j  & \text{if $n-j+1 > i > j+1$}\\	
					&= a_j  & \text{otherwise.}
				\end{eqnarray}
			\end{center}
		\end{minipage}
	\end{center}
	
	We claim that these $f_j$ form an HM-chain. By lines $4$ and $10$ of the definition of $f_i$, $f_1(x,y,y) = x$. By lines $2$ and $8$, $f_{n+1}(x,x,y) = y$. If $f_i(a_j,a_j,b_j) = a_j$, then $j+1 > i$ by line $1$, so $j+1 \geq i +1$, thus $f_{i+1}(a_j,b_j,b_j) = a_j$ by line $4$. Similarly, if $f_i(a_j,a_j,b_j) = b_j$, then $j+1 \leq i$ by line $2$ so, $j+1 < i+1$, and therefore $f_{i+1}(a_j,b_j,b_j) = b_j$ by line $3$. To sum up, $f_i(x,x,y) = f_{i+1}(x,y,y)$.
	
	It remains to show that each $f_i$ is arc preserving. In the first column of Table~\ref{pol_table}, we specify some $f_i(x,y,z)f_i(x',y',z')$ where $xx', yy', zz'$ are arcs, $x,y,z \subseteq \{a_j,b_j\}$ and $x',y',z' \subseteq \{a_{j+1},b_{j+1}\}$. We can assume that $f_i(x,y,z) = a_j$, since if $f_i(x,y,z) = b_j$, then $f_i(x,y,z)f_i(x',y',z')$ is an arc of $H$ because both $b_jb_{j+1}$ and $b_ja_{j+1}$ are arcs. It is straightforward to check that Table~\ref{pol_table} covers all cases.
	
	\begin{table}[h!tb]
		\caption{Cases in the proof of Lemma~\ref{ladder_pol}.}\label{pol_table}
		\begin{center}
	\noindent\begin{tabular}{ll}
		$f_i(a_j,b_j,b_j)f_i(a_{j+1},b_{j+1},b_{j+1})$ & \begin{minipage}{0.63\textwidth}Since $f_i(a_j,b_j,b_j) = a_j$ by assumption, $j+1 \geq i$ by line $4$. Therefore $j+2 \geq i$, so by line $4$, $f_i(a_{j+1},b_{j+1},b_{j+1}) = a_{j+1}$.\vspace{\gap}\end{minipage}\\
		\hline	
		$f_i(a_j,b_j,b_j)f_i(a_{j+1},a_{j+1},b_{j+1})$ & \begin{minipage}{0.63\textwidth}\vspace{\gap}We note that $j+1 \geq i$ as above, so $j + 2 > i$, and line $1$ gives that $f_i(a_{j+1},a_{j+1},b_{j+1}) = a_{j+1}$.\vspace{\gap}\end{minipage}\\
		\hline
		$f_i(a_j,b_j,b_j)f_i(a_{j+1},b_{j+1},a_{j+1})$  & \begin{minipage}{0.63\textwidth}\vspace{\gap}$j+1 \geq i \Rightarrow j + 2 \geq i$, so by line $11$ $f_i(a_{j+1},b_{j+1},a_{j+1}) = a_{j+1}$.\vspace{\gap}\end{minipage}\\
		\hline
		$f_i(b_j,a_j,b_j)f_i(a_{j+1},a_{j+1},b_{j+1})$  & \begin{minipage}{0.63\textwidth}\vspace{\gap}Since $f_i(b_j,a_j,b_j) = a_j$, $j+1 > i > n - j +1$ by line $5$, we have that $j+2 > i$, and therefore $f_i(a_{j+1},a_{j+1},b_{j+1}) = a_{j+1}$ by line $1$.\vspace{\gap}\end{minipage}\\
		\hline
		$f_i(b_j,a_j,b_j)f_i(b_{j+1},a_{j+1},a_{j+1})$  & \begin{minipage}{0.63\textwidth}\vspace{\gap}Since $f_i(b_j,a_j,b_j) = a_j$, $j+1 > i > n - j +1$ by line $5$, we have that $n - j < i$, and therefore $f_i(b_{j+1},a_{j+1},a_{j+1}) = a_{j+1}$ by line $9$.\vspace{\gap}\end{minipage}\\
		\hline
		$f_i(b_j,a_j,b_j)f_i(b_{j+1},a_{j+1},b_{j+1})$  & \begin{minipage}{0.63\textwidth}\vspace{\gap}Since $f_i(b_j,a_j,b_j) = a_j$, $j+1 > i > n - j +1$ by line $5$, we have that $j+2 > i > n - j$, and therefore $f_i(b_{j+1},a_{j+1},b_{j+1}) = a_{j+1}$ by line $5$.\vspace{\gap}\end{minipage}\\
		\hline
		$f_i(b_j,b_j,a_j)f_i(b_{j+1},b_{j+1},a_{j+1})$  & \begin{minipage}{0.63\textwidth}\vspace{\gap}Since $f_i(b_j,b_j,a_j) = a_j$, $n-j+1 \leq i$ by line $8$, we have that $n-j \leq i$, and therefore $f_i(b_{j+1},b_{j+1},a_{j+1}) = a_{j+1}$ by line $8$.\vspace{\gap}\end{minipage}\\
		\hline
		$f_i(b_j,b_j,a_j)f_i(b_{j+1},a_{j+1},a_{j+1})$  & \begin{minipage}{0.63\textwidth}\vspace{\gap}Since $f_i(b_j,b_j,a_j) = a_j$, $n-j+1 \leq i$ by line $8$, we have that $n - j < i$, and therefore $f_i(b_{j+1},a_{j+1},a_{j+1}) = a_{j+1}$ by line $9$.\vspace{\gap}\end{minipage}\\
		\hline
		$f_i(b_j,b_j,a_j)f_i(a_{j+1},b_{j+1},a_{j+1})$  & \begin{minipage}{0.63\textwidth}\vspace{\gap}Since $f_i(b_j,b_j,a_j) = a_j$, $n-j+1 \leq i$ by line $8$, we have that $n - j > i$ is violated in line $11$, so $f_i(a_{j+1},b_{j+1},a_{j+1}) = a_{j+1}$ by line $12$.\vspace{\gap}\end{minipage}\\
		\hline
		$f_i(b_j,a_j,a_j)f_i(b_{j+1},a_{j+1},a_{j+1})$  & \begin{minipage}{0.63\textwidth}\vspace{\gap}Since $f_i(b_j,a_j,a_j) = a_j$, $n-j+1 < i$ by line $9$, we have that $n - j < i$, so $f_i(b_{j+1},a_{j+1},a_{j+1}) = a_{j+1}$ by line $9$.\vspace{\gap}\end{minipage}\\
		\hline
		$f_i(a_j,b_j,a_j)f_i(a_{j+1},b_{j+1},a_{j+1})$  & \begin{minipage}{0.63\textwidth}\vspace{\gap}Since $f_i(a_j,b_j,a_j) = a_j$, $n-j+1 \leq i$ or $i \leq j + 1$ by line $12$, so we have that $n - j \leq i$ or $i \leq j + 2$. Therefore by line $12$, we have that $f_i(b_{j+1},a_{j+1},a_{j+1}) = a_{j+1}$.\vspace{\gap}\end{minipage}\\
		\hline
		$f_i(a_j,a_j,b_j)f_i(a_{j+1},a_{j+1},b_{j+1})$  & \begin{minipage}{0.63\textwidth}\vspace{\gap}Since $f_i(a_j,a_j,b_j) = a_j$, $j+1 > i$ by line $1$, so $j+2 > i$, and therefore line $1$ gives that $f_i(a_{j+1},a_{j+1},b_{j+1}) = a_{j+1}$.\vspace{\gap}\end{minipage}
	\end{tabular}
	\end{center}
	\end{table}
\end{proof} 
\section{A faster recognition algorithm}\label{faster}

An algorithm that recognizes digraphs containing no circular $\mathrm{N}$ is given in \cite{soda_lhom}. However, \cite{soda_lhom} only shows that the algorithm runs in polynomial time. In Appendix A, we show that a direct implementation of this algorithm when inputs are restricted to oriented trees is guaranteed to run in $O(|V(T)|^8)$ time. The running time of the algorithm in this paper is $O(|V(T)|^3)$.

\begin{theorem}\label{recognition_alg}
	Let $T$ be an oriented tree. Then there is a $O(|V(T)|^3)$ algorithm that decides whether $T$ contains a $\mathrm{Z_6}$ or a fuzzy $\mathrm{N}$ as induced subgraphs. (Equivalently, whether $T$ contains a circular $N$.)
\end{theorem}
\begin{proof}
	For each $u,v \in V(T)$, we find the unique oriented path from $u$ to $v$ in $T$ using, for example, breadth first search. For each such path $P(u,v)$, we test whether $P(u,v)$ is a $\mathrm{Z_6}$, or if not, we run the test in the next paragraph. Clearly, checking whether $P(u,v)$ is a $\mathrm{Z_6}$ takes constant time. If $\mathrm{Z_6}$ is found, we output YES.
	
	Otherwise, we run the following test on $P(u,v)$. Suppose that $P(u,v) = a_0 \dots a_n$ (where $a_0 = u$ and $a_n = v$). We traverse $P(u,v)$ from $a_0$ to $a_n$. We initialize a counter $c$ to $0$ at $a_0$, and increase and decrease $c$ every time we move from $a_i$ to $a_{i+1}$: if $a_ia_{i+1}$ is  forward arc, we increase $c$ by $1$, and if it is a backward arc, we decrease $c$ by $1$. If $c$ becomes negative at any step, then we abandon the computation on $P(u,v)$ (and we move on to analyze $P(u'v')$ for the next pair $u'v' \in V(T)$). As we traverse $P(u,v)$, we keep track of the maximum value $M$ of the counter $c$. Assume that we obtained the maximum value of $M$. If $M \leq 1$, then we abandon the computation on $P(u,v)$, and we move on to the next path $P(u',v')$. Otherwise $M \geq 2$. We set $c$ to $0$ again, and we traverse $P(u,v)$ again starting at $a_0$ as before, increasing and decreasing the value of $c$ at each step. We set $b_1 = a_0$. If $a_{k_1}$ is the first vertex where $c$ attains the value $M$, we set $t_1 = a_{k_1}$. If $a_{k_2}$ is the first vertex after $a_{k_1}$ where $c$ attains value $0$, then we set $b_2 = a_{k_2}$. If $a_{k_3}$ is the first vertex after $a_{k_2}$ where $c$ attains value $M$, then we set $t_2 = a_{k_3}$. If no such $a_{k_1}$, $a_{k_2}$ and $a_{k_3}$ exist, then we move on to the next path $P(u',v')$. Otherwise, we output YES.
	
	If we finished testing all paths $P(u,v)$ and we never output YES, then we output NO. This completes the description of the algorithm.
	
	If we output YES, then either $T$ contains a $\mathrm{Z}_6$, or the vertices $b_1,t_1,b_2,t_2$ for some $P(u,v)$ (defined above) satisfy the conditions of Lemma~\ref{fuzzy_N_present}, so $T$ contains a fuzzy $\mathrm{N}$.
		
		Conversely, suppose that $T$ contains a $\mathrm{Z}_6$ or a fuzzy $\mathrm{N}$ having first vertex $a$ and last vertex $b$. Then since we cycled through all pairs of vertices of $u,v \in V(T)$, if $T$ contains an induced $\mathrm{Z}_6$ with first vertex $a$ and last vertex $b$, the algorithm detects this $\mathrm{Z}_6$ when $u = a$ and $v = b$. Similarly, if $T$ contains an induced fuzzy $\mathrm{N}$ with first vertex $a$ and last vertex $b$, then the algorithm finds $b_1,t_1,b_2,t_2$ satisfying the conditions of Lemma~\ref{fuzzy_N_present} when $u = a$ and $v = b$.
		
		We cycle through all pairs $u,v \in V(T)$. For each such pair, we run a BFS to find $P(u,v)$, which takes time $O(|V(T)|)$ (recall that $T$ is an oriented tree). We traverse $P(u,v)$ and keep track only of a constant amount of data during these traversals. Overall the running time is $O(|V(T)|^3)$.
\end{proof}

\section{Conclusions and open problems}
We sharpened results of \cite{soda_lhom} regarding LHOM($H$) in the special case when $H$ is an oriented tree. The next natural question is how far can we push these results: Is there an inductive construction and forbidden subgraph characterization of digraphs $H$ for which LHOM($H$) in $L$? If an inductive characterization exists, can it be used to provide a simpler logspace algorithm for LHOM($H$)?
\bibliography{oriented_trees.bib}
\newpage
\section{Appendix A}\label{paths}

For the sake of completeness, we give two characterizations of oriented paths that contain no $\mathrm{Z_6}$ or fuzzy $\mathrm{N}$ as an induced subgraph. Let this class of oriented paths be denoted by $\cP$. The first characterization is a special case of the inductive characterization of oriented trees. The other characterization gives an explicit ``template'' such that all oriented paths in $\cP$ must obey this template. Roughly, this template specifies in what way certain fuzzy paths can be concatenated so that the resulting path is in $\cP$.

\subsection{Oriented Paths}

\begin{definition}\label{path_constructible}
	For $i \in \{1,2\}$, let $P_i$ be either the empty digraph, or an oriented path that has a single vertex $v_i$ in its top vertex level having outdegree $0$. The operation of taking the disjoint union of $P_1,P_2$ and a single new vertex $v_0$, and then adding arcs $v_1v_0$ (if $P_1$ is non-empty) and $v_2v_0$ (if $P_2$ is non-empty) to the resulting digraph is called taking the \emph{top up-join} of $P_1$ and $P_2$.
	
	Similarly, for $i \in \{1,2\}$, let $P_i$ be either the empty digraph, or an oriented path that has a single vertex $v_i$ in its bottom vertex level having indegree $0$. The operation of taking the disjoint union of $P_1,P_2$ and a single new vertex $v_0$, and then adding arcs $v_1v_0$ (if $P_1$ is non-empty) and $v_2v_0$ (if $P_2$ is non-empty) to the resulting digraph is called taking the \emph{bottom up-join} $P_1$ and $P_2$.

	We can similarly define the \emph{top down-join} and \emph{bottom down-join} of two oriented paths.
	
	If $P$ is an oriented path with a single vertex, we say that $P$ is \emph{constructible}. Inductively, if $P_1$ and $P_2$ are constructible oriented paths (possibly empty), then their top up-join, bottom up-join, top down-join, bottom down-join (when allowed) are also \emph{constructible}.
\end{definition}
Clearly, Definition~\ref{path_constructible} this is a special case of Definition~\ref{constructible}. Therefore this construction does not produce an oriented path that contains an induced \z6\ or a \fN\ as an induced subgraph. The converse is not difficult to prove, and a proof (using different terminology) can also be found in \cite{phd}.

\begin{theorem}
	An oriented path $P$ does not contain an induced $\mathrm{Z_6}$ or fuzzy $\mathrm{N}$ if and only if $P$ is constructible.
\end{theorem}

Now we give the template characterization.

\begin{definition}\label{wave}
	An oriented path $W$ is a \emph{wave} if $W$ is $\zz_k$ for some $1 \leq k \leq 5$, or $W$ is of the form $Q_1 A_1 Q_2 A_2 \dots Q_n A_n$, for some $n \geq 1$, where $A_i$ and $Q_i$ are defined as follows.	Let $P_1,\dots,P_n$ be minimal fuzzy paths such that for each $1 \leq i \leq n-1$, $height(P_i) >  height(P_{i+1})$. Then for each $1 \leq i \leq n$,
	\begin{itemize}
		\item if $i$ is odd, then $Q_i$ is of the form $P_i$, and $A_i$ is of the form $\zz_1$ or $\zz_3^{f=1}$, except that $A_n$ (if $n$ is odd) can be $\zz_j^{f=1}$ for any $j \leq 4$;	
		\item if $i$ is even, then $Q_i$ is of the form $\bar{P_i}$, and $A_i$ is $\zz_1$ or $\zz_3^{f=0}$, except that $A_n$ (if $n$ is even) can be $\zz_j^{f=0}$ for any $j \leq 4$.
	\end{itemize}
\end{definition}

\begin{theorem}\label{oriented_path_characterisation}
	Let $P$ be an oriented path that does not contain a $\mathrm{\zz_6}$ or a fuzzy $\mathrm{N}$ as an induced subgraph. Then $P$ has the form $\bar{U}AV$, or the form $r(\bar{U}AV)$, where $U$ and $V$ are waves, and $A$ is either of the form $\zz_1$ or $\zz_3^{f=0}$,
\end{theorem}

\begin{proof}
	Assume that $P = a_1\dots a_n$. \emph{To simplify notation, we write $P(i,j)$ instead of $P(a_i,a_j)$ (denoting the subpath of $P$ starting at $a_i$ and ending at $a_j$) in this proof.} We decompose $P$ as $\bar{U}AV$ or $r(\bar{U}AV)$ for some $U = Q_1' A_1' Q_2' A_2' \dots Q_{m'}' A_{m'}'$ and $V = Q_1 A_1 Q_2 A_2 \dots Q_m A_m$, and $A$, where we are using the notation in Definition~\ref{wave}.
	
	Let the levels of $P$ be $L_0,\dots,L_h$, and assume that $h \geq 2$, since otherwise $P$ is clearly a $\zz_k$ for some $1 \leq k \leq 5$. Let $P_1$ be a minimal fuzzy subpath of $P$ of maximum height. We observe first that there can be at most $2$ such paths. If there are more, we choose three arbitrary such subpaths $R_1, R_2, R_3$. Since $height(R_1) = height(R_2) = height(R_3) = h$, $R_1, R_2, R_3$ are edge-disjoint. So we can assume w.l.o.g\ that when we traverse $P$ from first to last vertex, we first traverse the arcs of $R_1$, then the arcs of $R_2$, and finally the arcs of $R_3$. Since we have three paths, we can find two among them such that they both have their first vertices in $L_0$, or they both have their last vertices in $L_h$. Assume w.l.o.g.\ that $R_1$ and $R_2$ are such paths. Set $b_1$ and $t_1$ be the first and last vertices of $P_1$, and $b_2$ and $t_2$ be the first and last vertices of $R_2$. Clearly, $b_1,t_1,b_2,t_2$ satisfy the conditions of Lemma~\ref{fuzzy_N_present}. This contradicts that $P$ contains no fuzzy $\mathrm{N}$ as an induced subgraph.
	
	Therefore we assume first that $P$ contains only one minimal fuzzy subpath of maximum height $h$. Let this subpath be $P_1 = P(i_1,j_1)$, where $i_1 < j_1$. We define $V$. We can assume w.l.o.g.\ that $P_1$ is an upward path, in which case we show that $P$ has the form $\bar{U}AV$. (If $P_1$ is a downward path, then we can work with $r(P)$ and proceed the same way: we show that $r(P)$ has the form $\bar{U}AV$, and therefore $P$ has the form $r(\bar{U}AV)$.)
	
	We set $Q_1 = P_1$. We find $i_2$ ($j_1 \leq i_2$) maximal such that $P(j_1,i_2)$ has height at most $1$, and $P(j_1,i_2)$ is either $Z_1$ or $Z_3^{f=0}$. We set $A_1$ to be $P(j_1,i_2)$. Then we find $j_2$ ($i_2 \leq j_2$) maximal such that $\bar{P}(i_2,j_2)$ is a minimal fuzzy path, and set $Q_2 = \bar{P}(i_2j_2)$. Notice that $\bar{P}(i_2,j_2)$ can be chosen to be minimal, since otherwise either $A_1$ was chosen to be a $\zz_1$ instead of $\zz_3^{f=0}$, or if $A_1$ is a $\zz_3^{f=0}$ and $\bar{P}(i_2,j_2)$ cannot be chosen to be minimal, then $P$ contains a $\zz_6$, a contradiction. Then we find $A_2$, similarly to the way we found $A_1$. We keep on defining $Q_i$ and $A_i$ this way until the following condition applies. For some $\ell$, $Q_{\ell} = P(i_\ell,j_\ell)$ (or $Q_{\ell} = \bar{P}(i_\ell,j_\ell)$) is defined, and the rest of $P$, i.e., $P(j_\ell, n)$ has height at most $1$ ($\star$). Then we set $A_\ell = P(j_\ell, n)$, and this completes the construction of $V$.	We show below that $height(Q_i) > height(Q_{i+1})$, which clearly that condition ($\star$) will eventually occur. Therefore the definition of $V$ is valid. Notice that $A_\ell$ is some $\zz$ on at most $4$ vertices, since otherwise $A_\ell$ together with the last arc of $Q_\ell$ would form a $\zz_6$.
	
	We already saw that the $A_i$ have the right form. To prove that $V$ is a wave, it remains to show that $height(Q_i) > height(Q_{i+1})$ for each $1 \leq i \leq m-1$. By the choice of $Q_1$ we have that $height(Q_1) > height(Q_2)$. Assume inductively that $height(Q_i) > height(Q_{i+1})$, and assume for contradiction that $height(Q_{i+1}) \leq height(Q_{i+2})$. As before, it is easy to use Lemma~\ref{fuzzy_N_present} to show the presence of a \fN\ in $P$, which gives a contradiction.

	If the subpath $a_{j_1-2}a_{j_1-1}a_{j_1}$ has height $1$, then we choose $A$ to be this subpath. Otherwise, $A$ is just $a_{j_1}$. $U$ is the subpath $a_t a_{t-1} \dots a_1$, where $t = j_1$ if $A$ is $a_{j_1}$, and $t = j_1-2$ otherwise. As for $V$, we can show that $U$ is a wave.
		
	Assume now that $P$ contains two minimal fuzzy subpaths $R_1$ and $R_2$ of maximum height $h$. Suppose w.l.o.g.\ that the arcs of $R_1$ appear before the arcs of $R_2$ in $P$. We also assume w.l.o.g.\ the $R_1$ is a downward path (if $R_1$ is upward path, we can work with $r(P)$ as above). Then $R_2$ must be an upward path such either $R_1R_2$ is a subpath of $P$, or $R_1AR_2$ is a subpath of $P$, where $A$ is a $\zz_3^{f=0}$. This $A$ corresponds to the $A$ in $\bar{U}AV$ the statement of the lemma. We define $V$ to be the subpath of $P$ starting with $R_2$ and ending at $a_n$, and $\bar{U}$ to be the subpath of $P$ starting at$a_1$ and ending with $R_1$. The proof that $U$ and $V$ are waves is similar to the proof of the first case above.
\end{proof}

\subsection{Running time analysis}

We give a running time analysis of the recognition algorithm in Theorem~6.8 of \cite{soda_lhom} when run an oriented tree. We call this algorithm the N-algorithm. Note that we did not attempt to optimize the implementation of this algorithm, and it is possible that the running time could be improved with some work.

The N-algorithm first produces a digraph $G^{++}$ with vertex set $|V(H)|^3$. If $H$ is a tree, then $G^{++}$ cannot have a cycle (see the definition of $G^{++}$ in \cite{soda_lhom}), so it must be a forest. It follows that $G^{++}$ has at most $O(|V(H)|^3)$ arcs. Then for each $x,y \in V(H)$, the N-algorithm finds the set $S_{(x,x,y)}$ of all vertices reachable from $(x,x,y)$ ignoring certain arcs of $G^{++}$, and finds the set $S_{(x,y,y)}$ of all vertices reachable from $(x,y,y)$ ignoring some other certain arcs of $G^{++}$. Using BFS, each such search could take $O(|V(H)|^3)$ time. The size of each $S_{(x,x,y)}$ and $S_{(x,y,y)}$ could be $O(|V(H)|^3)$. Then the N-algorithm checks if $S_{(x,x,y)} \cap S_{(x,y,y)}$ is non-empty, for each $x, y\in V(H)$. This can be done in time $O(|V(H)|^2 \cdot (|V(H)|^3)^2) = O(|V(H)|^8)$.
\section{Appendix B: alternative proofs}

\begin{lemma}\label{d2}
	If an oriented tree $T$ is constructible, then $T$ contains neither a $\mathrm{Z_6}$ nor a fuzzy $\mathrm{N}$ as an induced subgraph.
\end{lemma}
\begin{proof}
	Assume that $T$ is the up-join of $T_0,\dots,T_n$, and $T_0,\dots,T_n$ do not contain a $\mathrm{Z_6}$ or fuzzy $\mathrm{N}$ as an induced subgraph. Clearly, $T$ can contain a $\mathrm{Z_6}$ only in $\al_{\ell(v_0)-1}$. Note that by Definition~\ref{L_construction}, each $T_i$ contains at most one component in $\al_{\ell(v_0)-1}$, which is an out-star since $T_i$ contains at most one vertex in $L_{\ell(v_0)-1}$ with out-degree more than $0$. It follows that level $\al_{\ell(v_0)-1}$ of $T$ is an in-spider, and therefore it does not contain a $\mathrm{Z_6}$ as an induced subgraph.
	
	Suppose that $T$ contains a fuzzy $\mathrm{N}$ denoted by $N$. Let $P_1, \bar{P}_2, P_3$ the subpaths of $N$ as defined in Definition~\ref{def_fuzzyN}. Let $t_1$ be the last vertex of $P_1$, and $b_2$ be the first vertex of $P_2$. Let $v_0,v_1,\dots,v_n$ be the central and join vertices of $T$. Since $T_0,\dots,T_n$ do not contain a fuzzy $\mathrm{N}$, $N$ must contain an arc of the form $v_iv_0$, where we assume without loss of generality that $i=1$. It follows that $N$ must contain vertices both in $L_{\ell(v_0)}$ and $L_{\ell(v_1)}$.
	
	We show that $N$ cannot contain a vertex both in $L_{\ell(v_0)+1}$ and $L_{\ell(v_1)-1}$, contradicting that $N$ has height at least $2$. Assume first that $N$ contains a vertex in $L_{\ell(v_0)+1}$. By the definition of a fuzzy $\mathrm{N}$, then each of $P_1$, $\bar{P}_2$ and $P_3$ must contain a vertex in $L_{\ell(v_0)+1}$. Consider $j$ such that $t_1$ is a vertex of $T_j$. Suppose first that $j \neq 0$. Then both $P_1$ and $\bar{P}_2$ have their first vertex in $T_j$. Also, both $P_1$ and $\bar{P}_2$ must contain a vertex in $L_{\ell(v_1)}$. Since $v_j$ is the only vertex of $T_j$ in $L_{\ell(v_1)}$ with more than zero outneighbours in $T_j$, both $P_1$ and $\bar{P}_2$ must contain $v_j$. But then $T_j$ must contain a cycle, contradicting that $T_j$ is a tree. If $j = 0$, then since $v_0$ is the only vertex of $T_0$ with inneighbours in $L_{\ell(v_1)}$, both $P_1$ and $\bar{P}_2$ must contain a $v_0$, and thus $T_0$ contains a cycle. This is impossible.
		
	Assume therefore that $N$ contains a vertex in $L_{\ell(v_1)-1}$. We find $j$ such that $b_2$ is a vertex of $T_j$. As above, it follows that both $\bar{P}_2$ and $P_3$ contain vertex $v_j$, indicating that $T_j$ contains a cycle. This is a contradiction.
	
	The analysis is analogous when $T$ is obtained using a down-join operation.	
\end{proof}

The proof of the following lemma is easy, or it can be extracted from \cite{soda_lhom}.
\begin{lemma}\label{imp_prop}
	Let $H$ be a digraph. Then if $H$ contains a circular $\mathrm{N}$ with congruent walks $X = x_0 x_1\dots x_n$, where $x_0=x_n=x$, $Y = y_0 y_1 \dots y_n$, where $y_0 = y_n = y$, and $Z = z_0 z_1 \dots z_n$, where $z_0 = y$ and $z_n = x$, then $H$ has the following \emph{implication property}. Let $P = p_0 p_1 \dots p_n$ be an oriented path congruent to $X$ (and hence also to $Y$ and $Z$) with lists $L(p_i) = \{x_i, y_i, z_i\}$, for each $0 \leq i \leq n$. Then there are list homomorphisms $\varphi_{xx}$, $\varphi_{yy}$, and $\varphi_{yx}$, each from $P$ to $H$ such that
	\begin{itemize}
		\item $\varphi_{xx}(p_0) = \varphi_{xx}(p_n)=x$,
		\item $\varphi_{yy}(p_0) = \varphi_{yy}(p_n)=y$,
		\item $\varphi_{yx}(p_0)=y$ and $\varphi_{yx}(p_n)=x$,
	\end{itemize}
	and there is no list homomorphism $\varphi_{xy}$ from $P$ to $H$ such that $\varphi_{xy}$ maps $p_0$ to $x$ and $p_n$ to $y$.
\end{lemma}

\begin{lemma}\label{cN->z}
	If an oriented tree $T$ contains a circular $\mathrm{N}$, then it contains a $\mathrm{Z_6}$ or a fuzzy $\mathrm{N}$ as an induced subgraph.
\end{lemma}
\begin{proof}
	We show that if $T$ contains neither an induced $\mathrm{Z_6}$ nor a fuzzy $\mathrm{N}$, then $T$ cannot have the implication property. This implies by Lemma~\ref{imp_prop} that $T$ cannot have a circular $\mathrm{N}$.
	
	Since we are assuming that $T$ has no induced $\mathrm{Z_6}$ or fuzzy $\mathrm{N}$, $T$ is the up-join (or down-join) of some trees $T_0,\dots,T_n$. We suppose that the central vertex is $v_0$, and the join vertices are $v_1,\dots,v_n$. We assume that $T$ is the up-join and note that the analysis for down-join is similar. We inductively assume that we already showed that none of $T_0,\dots,T_n$ contains a circular $\mathrm{N}$ (since they do not contain an induced $\mathrm{Z_6}$ or fuzzy $\mathrm{N}$). Therefore the fuzzy $\mathrm{N}$ in $T$ must contain an arc $v_iv_0$ for some $1 \leq i \leq n$. Let $X$ (with first vertex $x$), $Y$ (with first vertex $y$) and $Z$ be the congruent walks making up the circular $\mathrm{N}$ in $T$, and let $P = p_0 \dots p_n$ be a path congruent to $X$. Let $\varphi_{xx}$, $\varphi_{yy}$, and $\varphi_{yx}$ be the list homomorphisms associated with the implication property from Lemma~\ref{imp_prop}. Informally, all arguments below will use the fact that $v_0$ is a ``bottleneck vertex'' in $T$.
	
	Suppose first that $x \in V(T_{j_x})$ and $y \in V(T_{j_y})$. \emph{Assume that $j_x \neq j_y$}. Let $p_\alpha$ and $p_{\beta}$ be vertices of $P$ such that
	\begin{itemize}
		\item $\alpha$ is minimum such that $\varphi_{yx}(p_\alpha) = v_0$,
		\item if $p_{\alpha}p_{\alpha + 1}$ is a forward arc,
		\begin{itemize}
			\item then let $\beta$ be minimum such that $\alpha < \beta$, $\varphi_{yx}(p_\beta) = v_0$, and $p_{\beta}p_{\beta+1}$ is a backward arc
			\item otherwise set $\beta = \alpha$.
		\end{itemize}
	\end{itemize}
	It is not difficult to see that such $\alpha$ and $\beta$ must exist. This is because $P$ is a path, so the image of $P$ in $T$ must contain a path between $y$ and $x$, and such a path must enter $T_0$ through $v_0$ and it must also leave $T_0$ through $v_0$. Observe that it follows from the above definition that the image of $P(p_\alpha,p_\beta)$ under $\varphi_{yx}$ is entirely in $T_0$.
	
	We modify $\varphi_{xx}$ as
	\[\varphi_{xx}'(u)=
	\begin{cases}
	\varphi_{xx}(u) \text{ if $u$ is not a vertex of $P(p_\alpha,p_\beta)$}\\
	\varphi_{yx}(u) \text{ if $u$ is a vertex of $P(p_\alpha,p_\beta)$}
	\end{cases}
	\]
	We claim that $\varphi_{xx}'$ is a list homomorphism from $P$ to $T$. Clearly, we only need to check that arcs $p_{\alpha-1}p_{\alpha}$ and $p_{\beta}p_{\beta+1}$ are mapped to an arc of $T$. Since $p_{\alpha}$ has indegree $1$ (by minimality of $\alpha$), and $p_{\alpha}$ is mapped to a vertex in level $L_{\ell(v_0)}$, $p_{\alpha-1}$ must be mapped to $v_i$ for some $1 \leq i \leq n$. (Recall that $v_1,\dots,v_n$ are the only vertices of $T_1,\dots,T_n$ in level $L_{\ell(v_0-1)}$ with outdegree greater than zero.) So $\varphi_{xx}'(p_{\alpha-1}) = v_i$ and $\varphi_{xx}'(p_{\alpha}) = v_0$, and thus $p_{\alpha-1}p_{\alpha}$ is mapped to an arc of $T$. An analogous argument shows that $p_{\beta}p_{\beta+1}$ is also mapped to an arc. In addition, we can show in a similar way that we can construct $\varphi_{yy}'$ that agrees with $\varphi_{yy}$ everywhere, except that it takes the value of $\varphi_{yx}$ on vertices of $P(p_\alpha,p_\beta)$.
	
	Now we can construct $\varphi_{xy}$, the list homomorphism that maps $p_0$ to $x$ and $p_n$ to $y$ showing that $T$ does not have the implication property. The function $\varphi_{xy}$ on $P(p_0,p_{\alpha-1})$ is $\varphi_{xx}'$.  On $P(p_\alpha, p_\beta)$, $\varphi_{xy}$ is $\varphi_{yx}$, and on $P(p_{\beta+1},p_n)$, $\varphi_{xy}$ is $\varphi_{yy}'$. Since $\varphi_{xx}'(p_\alpha) = \varphi_{yx}(p_\alpha) = v_0$, and $\varphi_{yx}(p_\beta) = \varphi_{yy}'(p_n) = v_0$, $\varphi_{xy}$ is a list homomorphism from $P$ to $T$, giving a contradiction.
	
	Assume therefore that $j_x = j_y$. Then if $\varphi_{yx}$ maps any vertex of $P$ to $v_0$, then we can use the same argument as above.  If $\varphi_{yx}$ does not map any vertex of $P$ to $v_0$, then at least one of $\varphi_{xx}$ or $\varphi_{yy}$ must map a vertex of $P$ to $v_0$, since otherwise there would be a circular $\mathrm{N}$ entirely in $T_{j_x }$, contradicting the induction hypothesis. A similar argument can be used to construct a list homomorphism $\varphi_{xy}$ mapping $p_0$ to $x$ and $p_n$ to $y$ from $\varphi_{xx}$ and $\varphi_{yy}$. The existence of $\varphi_{xy}$ gives a contradiction.
		
	The remaining cases can be analyzed in a very similar way, and therefore we give only brief arguments. Assume that $x,y \in T_0$. Let $\varphi$ be any list homomorphism associated with the implication property. Then we have that $\varphi(p_0), \varphi(p_n) \in \{x,y\}$.
	
	Let $\alpha$ and $\beta$ be such that
	\begin{itemize}
		\item $\alpha$ is minimum such that $\varphi(p_{\alpha}) \in L_{\ell(v_0)}$, and $p_{\alpha} p_{\alpha+1}$ is a backward arc. (Note that such an $\alpha$ must exist since otherwise there is a \cN\ entirely in $T_0$.)
		\item $\beta$ is maximum such that $\varphi(p_\beta) \in L_{\ell(v_0)}$, and $p_{\beta-1} p_\beta$ is a forward arc.
	\end{itemize}
	Note that $\varphi$ (any of the list homomorphisms associated with the implication property) maps $p_{\alpha}$ and $p_{\beta}$ to $v_0$. Therefore let $\varphi_{xy}$ be the function that is $\varphi_{xx}$ on vertices of $P(p_0,p_\beta)$, and $\varphi_{yy}$ on vertices of $P(\beta+1,n)$. Since $\varphi_{xx}(p_\beta) = v_0 = \varphi_{yy}(p_\beta)$, $\varphi_{xy}$ is a list homomorphism, and we have the desired contradiction.
	
	Assume next that $x \in V(T_0)$ and $y \in V(T_q)$ for some $1 \leq q \leq n$. (The case when $x \in V(T_0)$ and $y \in V(T_q)$ can be analyzed similarly.) We define $\alpha$ and $\beta$ as in the previous case. When $\varphi$ (above) is $\varphi_{xx}$, $\varphi_{xx}(p_{\alpha}) = \varphi_{xx}(p_{\beta}) = v_0$. Then let $\varphi_{xy}$ be the function that is $\varphi_{xx}$ on vertices of $P(p_0,p_\beta)$, and $\varphi_{yy}$ on vertices of $P(p_{\beta+1},p_n)$. As before, we can check that $\varphi_{xy}$ is a list homomorphism, and we have the desired contradiction.
\end{proof}
	
	\begin{lemma}\label{no_HM}
		Let $T$ be an oriented tree. If $T$ contains a $\mathrm{Z_6}$ or a fuzzy $\mathrm{N}$ as an induced subgraph, then $T$ does not have conservative polymorphisms $f_1$, $f_2$ and $f_3$ that form a Hagemann-Mitschke chain.
	\end{lemma}
	\begin{proof}
		If $T$ contains a $\mathrm{Z_6}$ or a fuzzy $\mathrm{N}$ as an induced subgraph, then $T$ contains a circular $\mathrm{N}$ by Theorem~\ref{circ_N_ind}. If $T$ contains a circular $\mathrm{N}$, then $T$ does not have conservative polymorphisms that form a Hagemann-Mitschke chain of any length by \cite{soda_lhom}.
	\end{proof} 
\end{document}